\documentclass[cmp]{svjour_no_cmp}

\usepackage{latexsym}
\usepackage{graphics}
\usepackage{hyperref, enumerate}
\usepackage{soul}
\usepackage{graphicx}
\usepackage{amsmath}
\usepackage{amssymb}
\usepackage{verbatim}
\usepackage{color}
\usepackage{cancel}

\def\Tr{\mathop {\rm Tr}}

\def\anti#1#2{\left\{#1,#2\right\}}
\def\comm#1#2{\left[#1,#2\right]}

\def\calfh{\widetilde{\mathcal H}}

\def\fh{\widetilde H} 
\def\be{\begin{equation}}
\def\beq{\begin{eqnarray}}
\def\ee{\end{equation}}
\def\eeq{\end{eqnarray}}
\def\eqref#1{(\ref{#1})}
\def\lra#1{\langle #1 \rangle}
\def\lrp#1{\left( #1 \right)}

\def\abs#1{\left\vert #1\right\vert}

\def\A{\mathfrak{A}}

\def\even{{\rm even}}
\def\nn{\nonumber\\}
\def\Tr{ {\rm Tr}}

\def\anti#1#2{\left\{#1,#2\right\}}

\def\calfh{\widetilde{\mathcal H}}

\def\fh{\widetilde H} 
\def\be{\begin{equation}}
\def\ee{\end{equation}}
\def\beq{\begin{eqnarray}}
\def\eeq{\end{eqnarray}}
\def\eqref#1{(\ref{#1})}
\def\lra#1{\langle #1 \rangle}
\def\lrp#1{\left( #1 \right)}
\def\A{\mathfrak{A}}
\def\H{\widetilde{H}}
\def\HH{\widetilde{\mathcal{H}}}
\def\comm#1#2{\left[#1,#2\right]}

\renewcommand{\thesection}{\Roman{section}}
\renewcommand{\thesubsection}{\thesection.\arabic{subsection}}

\newtheorem{thm}{Theorem}%[section]
\newtheorem{cor}[thm]{Corollary}

\newtheorem{prop}[thm]{Proposition}
\newtheorem{conj}[thm]{Conjecture}
\newtheorem{defn}[thm]{Definition}

\usepackage{graphics}
\begin{document}
\title{Vortex Loops and Majoranas}
\author{Stefano Chesi\inst{1,2} \and Arthur Jaffe\inst{3,4,5} \and Daniel Loss\inst{4,2} \and Fabio L. Pedrocchi\inst{4}% etc
}                     % Do not remove
\institute{Department of Physics, McGill University, Montreal, Quebec, Canada H3A 2T8 \and CEMS, RIKEN, Wako, Saitama 351-0198, Japan \and Harvard University, Cambridge, Massachusetts 02138, USA \and Department of Physics, University of Basel, Basel, Switzerland \and Institute for Theoretical Physics, ETH Z\"urich, Z\"urich, Switzerland}

% Add name of the expert who has communicated your paper
%\communicated{name}
%
\maketitle
\begin{abstract}
We investigate the role that vortex loops play in characterizing eigenstates of  interacting Majoranas.  We first give some general results, and then we focus on ladder Hamiltonian examples to test further ideas. Two methods yield exact results:  i.)  We utilize the mapping of spin Hamiltonians to quartic interactions of Majoranas and show under certain conditions the spectra of these two examples coincide.   ii) In cases with reflection-symmetric Hamiltonians, we use reflection positivity for Majoranas to characterize vortices.   Aside from these exact results, two additional methods suggest wider applicability of these results:  iii.) Numerical evidence suggests similar behavior for certain systems without reflection symmetry.  iv.) A perturbative analysis also suggests similar behavior without the assumption of reflection symmetry.  
\end{abstract}
\setcounter{equation}{0}

\section{Introduction}
\subsection{Motivation}
The spin systems we study here have their origin in the ``compass model'' introduced by Kugel and Khomskii~\cite{KugelKhomskii} to study the Jahn-Teller effect in magnetic insulators. This model has a rich structure, and in a two-dimensional version known as the honeycomb model, Kitaev found excitations with fractional statistics   in their solution \cite{KitaevHoney}.  There is a long history of fractional-statistics excitations arising from braid statistics and now called anyons,  see for example \cite{Streater-Wilde,Goldin1971,LM,WilczekPRL,FM,FRS,Longo}.

In the honeycomb model (possibly with a magnetic field) both abelian and non-abelian anyons occur for different values of the coupling constants. There has also been extensive study of ``ladders" with possibly anyonic excitations \cite{Vishveshwara2011,Chesi2012,Chen2013}.

Models with anyons appear of interest in current studies of topological quantum computing. One argues that the degenerate ground-state subspace of such systems is a good place to store and to process quantum information, see for example \cite{KitaevToric,KitaevBook,Bravyi2002,FLW,Lloyd,FKLW,Preskill,Pachos}. The advantage for storage is that the subspace of ground states is stable against a wide class of local perturbations. On the other hand, non-local perturbations (such as braiding of anyons) may implement quantum gates.

In this context, it is important to understand the properties of the ground states, for these states are candidates to encode quantum information. In particular, we are interested in the question whether vortices are present or absent in the ground states, which is intimately connected to our study of Majoranas. In particular for a large system, the location of the vortex should be independent of its energy, and therefore labels a degeneracy. Transitions between the different ground states could be counter-productive to storage of information.

We take advantage here of the equivalence of a quartic Hamiltonian $\H$ of interest,  to a family $\{\H_{u}\}$ of Hamiltonians describing quadratic interactions of Majoranas. One naturally comes to the question:  which Hamiltonian within the family has the lowest ground-state energy? The answer to this question characterizes the ground states of the original Hamiltonian.
Besides being of fundamental interest in its own right, the question of determining which Hamiltonian has minimal ground-state energy turns out to be related to the presence or absence of vortices, and hence to the storage of information as outlined above. 
Although other methods might allow one to study the spin Hamiltonians directly, we have not investigated that possibility.

\subsection{Goals and Results}
In this paper we explore some general properties of a family of Majorana interactions on a cubic lattice in $d$ dimensions. We focus on the question of how to identify which Hamiltonians within the family have minimal ground-state energy.  

Our first approach is to use reflection positivity to study the ground states of a family of Hamiltonians that are quadratic in Majoranas.   In the case of charged excitations, such questions were studied in \cite{Lieb1994,Nachtergaele1996}. However here we study neutral systems.  We use the results of \cite{JP} to establish for reflection-symmetric Hamiltonians  that the ground states of the Hamiltonians in the family with the lowest ground-state energy are vortex free.  We give the precise statement in Theorem~\ref{thm:main_theorem}. 

As Osterwalder and Seiler pointed out in their original study of the Wilson action on a lattice \cite{OS3}, when one studies reflection positivity in gauge theory, a useful technique is to perform a unitary gauge transformation which removes the interaction terms across the reflection plane.  Our Hamiltonians  have a gauge symmetry as well, which allows one to fix the sign of the interactions across the reflection plane in our proof of   
Theorem \ref{thm:main_theorem}, as it also was in  \cite{Lieb1994}.

The  Majorana interactions considered here arise naturally from  mapping certain nearest-neighbor quadratic interactions of quantum spins $H$ on a trivalent lattice into quartic interactions of Majoranas  $\widetilde H$ \cite{KitaevHoney}.  In  \S\ref{sec:Exact_Results}, Theorem \ref{prop:multiplicity}, we prove that for open boundary conditions, the spectra of these two Hamiltonians $H$ and $\widetilde  H$ are the same (except for multiplicity).  
On the other hand, using numerical methods we show that the spectra are different  in the case of  periodic boundary conditions.  In spite of the fact that in general the eigenvalues of  $H$ and $\widetilde H$ do not coincide,  numerical evidence suggests that the ground-state energies of these two Hamiltonians are the same. 

Our next approach is to apply reflection-positivity in order to characterize vortex loop configurations of the ground states for fermionic systems arising from certain spin ladders. We use reflection positivity, proved in \cite{JP}, and in Theorem \ref{prop:Vortex_Closed_Spin_Ladder} we show that in the case of a reflection-symmetric interactions with  couplings of the same sign, the ground state is vortex free.  In particular, application of this new result to the specific example of ladder spin systems that we study here in detail is of interest, see Corollary \ref{rem:vortex-free}. 

In the general case when reflection-symmetry is broken, we investigate the properties of the ground-state energy using numerical methods to compute low-lying energy levels.  This investigation indicates that the ground state remains vortex free (with the relevant restriction on the sign of various couplings) even though our mathematical proof does not apply.  Based on this information, we formulate a general {\em vortex-free ground-state conjecture} for ladders in \S\ref{sec:conjecture}.
We also study spin ladders by perturbation theory.  We show that  the ground state remains vortex free, and we incorporate these insights into the conjecture in \S\ref{sec:conjecture}.

\subsection{Organization of the Paper}
In \S\ref{sec:Majorana_Cubic_Lattice} we  define   a family of Hamiltonians with nearest-neighbor Majorana interactions on a cubic lattice in arbitrary dimension.   In this section we assume the existence of a reflection plane leaving the lattice invariant and transforming the Hamiltonians in a simple way. Using reflection-positivity 
one can characterize vortex loop configurations of the Hamiltonians   that minimize the ground state energy within the given family. When all the coupling constants are positive (or negative), the minimal-energy is achieved for a vortex-free ground state. This property of vortex loops is related to results of Lieb \cite{Lieb1994}, and of Macris and Nachtergaele \cite{Nachtergaele1996} for hopping Hamiltonians.

In \S\ref{sec:Quantum_Spin_Ladders}--\S\ref{sect:Ladders and Reflections} we apply these results to spin ladders and their Majorana fermionic representations. 
While in much of this paper we analyze ladders as an example, most of our results extend in a straightforward way to models defined on a honeycomb lattice with similar trivalent couplings at each site. 

In \S\ref{sec:Exact_Results} we show that the spectrum of an open spin ladder coincides with the spectrum of its Majorana fermionic representation, aside from multiplicity. While the spectrum of a closed spin ladder seems not to have this property, we conjecture that the ground state energies are the same.

In \S\ref{sec:Numerical_Evidence} we study certain ladders numerically. These ladders do not possess the symmetry required to use reflection-positivity arguments. Numerical evidence suggests that the ground state energy of a closed spin ladder coincides with the ground state energy of its fermionic representation. Furthermore, the numerical calculations suggest that the ground states remain vortex-free (or vortex-full) as for the spin ladders for which reflection-positivity applies. 

In \S\ref{sec:non reflection-symm} we use  third-order perturbation theory (the lowest non-trivial order) to complement the picture. These results also show that for certain regions of the coupling constants for non-symmetric, open and closed ladders, the ground states are vortex-free (or vortex-full).

\setcounter{equation}{0}
\section{Nearest-Neighbor Majorana Interactions on a Cubic Lattice}\label{sec:Majorana_Cubic_Lattice}
\subsection{The Cubic Lattice}
We consider a finite subset $\Lambda$ of the cubic lattice $\mathbb{Z}^{d}$ in Euclidean $d$-space, with an even number $\abs{\Lambda}$ of sites $i$. We assume $\Lambda$ to be a rectangular box, with sites $i\in\mathbb{Z}^{d}$ and bonds $(ij)$ connecting nearest-neighbor sites. The side length of the box along each coordinate axis may be different. We call this an \emph{open} box. We sometimes \emph{close} the box in one or more coordinate directions. One closes the box in the $k^{\text{th}}$ direction by defining sites with minimum and maximum value of the $k^{\text{th}}$ coordinate, but the same value of each of the other coordinates, to be nearest-neighbors. 

\subsection{The Majoranas and The Hilbert Space}
A set of Majoranas is a self-adjoint representation of an even-dimensional Clifford algebra,  
	\be
	\anti{c_{i}}{c_{j}}=2\delta_{ij}\;,
	\qquad\text{where}
	\quad
	c_{j}=c_{j}^{*}=c_{j}^{-1}\;.
	\ee
We assign a Majorana $c_{j}$ to each site $j$. 
Majoranas can be represented on a Fock-Hilbert space $\HH_{c}$ of dimension $2^{\abs{\Lambda}/2}$ and we use this representation.
We consider the family of Hamiltonians 
	\be\label{eq:u_Hamiltonian}
	\H_{u}=\sum_{(ij)}J_{(ij)}\,u_{ij}\,i\,c_{i}c_{j}\,,
	\ee
with $J_{(ij)}=J_{(ji)}\geqslant0$ and $u_{ij}=-u_{ji}=\pm1$. In case the subscripts are difficult to distinguish, we write $J_{(i,j)}$ in place of $J_{(ij)}$.
\subsection{Vortex Loops}
Define a loop $\mathfrak{C}$ of length $\abs{\mathfrak{C}}=\ell$ as an ordered sequence of nearest-neighbor sites $\{i_{1},i_{2},\ldots,i_{\ell},i_{1}\}$ in $\Lambda$, starting and ending at the same site. In addition we assume $i_{1},\ldots,i_{\ell}$ are distinct so the loop is not self-intersecting. We identify the loop with a closed, directed path connecting nearest-neighbor sites $i_{k}$ and $i_{k+1}$ by bonds $(i_{k}\,i_{k+1})$. Denote $-\mathfrak{C}$ as the reverse loop  which contains the same sites as $\mathfrak{C}$ but the opposite orientation, $\{i_{1},i_{\ell},i_{\ell-1},\ldots,i_{2},i_{1}\}$. Let $\prod_{(ij)\in\mathfrak{C}}K_{ij}$ denote the ordered product around the loop,
	\be
	\prod_{(ij)\in\mathfrak{C}}K_{ij}=K_{i_{1}\,i_{2}}\,K_{i_{2}\,i_{3}}\cdots K_{i_{\ell-1}\,{i_{\ell}}}\,K_{i_{\ell}\,i_{1}}\,.
	\ee  
In the case where $K_{ij}$ are matrices, the starting point of the loop is important, though the trace $\Tr\lrp{\prod_{(ij)\in\mathfrak{C}}K_{ij}}$ is independent of the cyclic permutation of sites in the loop. The smallest loop contains four sites, which are the corner of a square or plaquette $p$ bounded by the loop $\mathfrak{C}=\partial p$. Define a loop to be {\em non-degenerate} if the coupling constants on the loop do not vanish:
	\be
		\mathfrak{C} \text{ is non-degenerate } 
		\Leftrightarrow
		\prod_{(ij)\in\mathfrak{C}}J_{(ij)}\neq0\,.
	\ee
Define the vortex loop $\widetilde{\mathfrak{B}}(\mathfrak{C})$ as
	\be\label{eq:vortex_loop}
		\widetilde{\mathfrak{B}}(\mathfrak{C})
		=-\prod_{(ij)\in\mathfrak{C}}\,u_{ij}\,.
	\ee
	
In case $\widetilde{\mathfrak{B}}(\mathfrak{C})=1$ we say that the loop $\mathfrak{C}$ is {\em vortex-free}. In case $\widetilde{\mathfrak{B}}(\mathfrak{C})=-1$ we say that $\mathfrak{C}$ is {\em vortex-full}.  We say that a state is vortex-free or vortex-full, in case all loops $\mathfrak{C}$ are vortex-free or vortex-full.   In case $\mathfrak{C}$ bounds a surface, one can interpret the vortex configuration $\widetilde{\mathfrak{B}}(\mathfrak{C})$ in terms of  flux through the surface. 

\subsection{Fermionic Fock Representation}
We represent the Hilbert space $\HH_{c}$ as a fermionic Fock space generated by $\abs{\Lambda}/2$  real creation operators $a_{\mu}^{*}$ and their adjoints $a_{\mu}$  are the corresponding annihilation operators. Here $\mu=1,\ldots,\abs{\Lambda}/2$. Each creation-annihilation pair gives rise to two Majoranas
	\be\label{eq:representation_Majorana}
		m_{\mu1}=a_{\mu}+a_{\mu}^{*}\,,
		\qquad\text{and}\quad
		m_{\mu2}=i\lrp{a_{\mu}-a_{\mu}^{*}}\,.
	\ee

\subsection{The $\mathbb{Z}_{2}$ Gauge Group on $\HH_{c}$}
It is convenient to introduce the gauge group $\mathfrak{G}^{c}$ that acts on $\HH_{c}$. The generators of this group are the operators
	\be\label{eq:gauge_generators_on_the_c}
		U^{c}_{j}
		=c_{j}\;\mathcal{U}^{c}\,,\qquad\text{where}\quad\mathcal{U}^{c}={\prod_{j=1}^{4N}}\,c_{j}\,.
	\ee
We later choose an order for the product $\mathcal{U}^{c}$, but conjugation by $\mathcal{U}^{c}$ does not depend on the choice. The group $\mathfrak{G}^{c}$ has dimension $2^{\abs{\Lambda}+1}$, since $\lrp{U_{j}^{c}}^{2}=-I$.

A general gauge transformation $W\in\mathfrak{G}^{c}$ on $\HH_{c}$ depends upon $\abs{\Lambda}+1$ two-valued parameters $\mathfrak{n}=\{n_{0},n_{1},\ldots,n_{\abs{\Lambda}}\}$. It has the form
	\be\label{eq:gauge_V}
	 W(\mathfrak{n})
	 =(-1)^{n_{0}}\lrp{U_{1}^{c}}^{n_{1}}\,\lrp{U_{2}^{c}}^{n_{2}}\,\cdots \lrp{U_{\abs{\Lambda}}^{c}}^{n_{\abs{\Lambda}}}\,,
	\ee
where $n_{k}=0,1$.
Conjugation by the unitary $W(\mathfrak{n})$ acts on the $c_{k}$'s as an automorphism that we also denote by $W(\mathfrak{n})$. We write,
	\be
	W(\mathfrak{n})(c_{k})
	=W(\mathfrak{n})\,c_{k}\,W(\mathfrak{n})^{*}
	=(-1)^{n_{k}} c_{k}\,.
	\ee

\subsection{Reflection-Symmetry}
In certain sections we consider lattices that are symmetric under a reflection $\vartheta$ in a hyperplane $\Pi$, that intersects no lattice sites. The reflection defines two disjoint subsets of the lattice $\Lambda_{\pm}$ of $\Lambda=\Lambda_{-}\cup\Lambda_{+}$ that map into each other,
	\be
		\vartheta\,\Lambda_{\pm}
		=\Lambda_{\mp}\,,\qquad 				\vartheta^2={\rm Id},\qquad\text{with}\quad \vartheta: i\mapsto \vartheta i
	\,.
	\ee
The reflection $\vartheta$ acts on loops as 
	\be
%	\vartheta(\mathfrak{C})
%	=\vartheta(\{i_{1},i_{2},\ldots,i_{\ell},i_{1}\})
%	=\{\vartheta i_{1},\vartheta i_{\ell},\vartheta i_{\ell-1},\ldots,	\vartheta i_{2},\vartheta i_{1}\}\,.
		\vartheta(\mathfrak{C})=\vartheta(\{i_{1},i_{2},\ldots,i_{\ell},i_{1}\})
		=\{\vartheta i_{1}, \vartheta i_{2},\ldots,\vartheta i_{\ell},\vartheta i_{1}\}	\;.
		\ee
We say that a loop $\mathfrak{C}$ is reflection-symmetric under the action of $\vartheta$, if $\vartheta(\mathfrak{C})=-\mathfrak{C}$.

We represent $\vartheta$ on $\HH_{c}$ as an \emph{anti-unitary} transformation with
%	\be\label{eq:alpha_j}
%	 	\vartheta(c_{j})
%		=\vartheta\,c_{j}\,\vartheta^{-1}
%		=\overline{c_{\vartheta j}}
%		=\alpha_{\vartheta j}\, c_{\vartheta j}\,.
%	\ee
	\be\label{eq:alpha_j}
	 	\vartheta(c_{j})
		=\vartheta\, c_{j}\, \vartheta^{-1}
		=c_{\vartheta j}\,.
	\ee
The transformation $\vartheta$ defines an anti-linear automorphism of the algebra generated by the $c_{j}$'s, which we also denote by $\vartheta$.
\begin{defn}
The Hamiltonian $\H_{u}$ is reflection-symmetric if $\vartheta(\H_{u})=\H_{u}$.
\end{defn}

\subsection{The Fermionic Algebra on $\HH_{c}$}
Define the fermionic algebra $\mathfrak{A}_{c}$ as the algebra generated by the $c_{j}$'s for $j\in\Lambda$. Let $\A_{c}^{\text{even}}$ denote the even subalgebra of $\mathfrak{A}_{c}$, generated by even monomials in the fermionic operators. Similarly let $\A_{c,\pm}\subset\mathfrak{A}_{c}$ denote the subalgebras generated by the $c_{j}$'s with $j\in\Lambda_{\pm}$. Also let $\mathfrak{A}_{c,\pm}^{\text{even}}$ denote the even subalgebras of $\mathfrak{A}_{c,\pm}$.

\subsection{Reflection Positivity}
Reflection positivity  (RP) for Majoranas is a condition on a Hilbert space, an algebra of operators on the Hilbert space, a reflection $\vartheta$ through a plane $\Pi$, and a Hamiltonian. Here we study the Hilbert space $\HH_{c}$, the algebras $\A^{\even}_{c,\pm}$, an implementation of the reflection $\vartheta$ on $\HH_{c}$, and a reflection-symmetric Hamiltonian $\H$. The RP condition states that
	\be\label{eq:reflection_positivity_functional}
		\Tr_{\HH_c}\lrp{B\,\vartheta(B)\,e^{-\H}}\geqslant0\,,\qquad\text{for all $B\in\A^{\even}_{c,\pm}$}\,.
	\ee

Time-reflection positivity was originally discovered in quantum field theory by Osterwalder and Schrader in the context of relating classical fields with quantum fields \cite{OS1}. In particular they introduced the method of ``multiple reflection bounds,'' involving iterated applications of a reflection-positivity bound.  Such bounds have been key for the first mathematical proof of the existence of phase transitions (ground-state degeneracy) in quantum field theory \cite{GJS}, and in proving that certain field theories have infinite volume limits \cite{GJ}. 

RP has also had many applications in the study of phase transitions for classical and quantum spin systems on a lattice; see  Fr\"ohlich,  Simon, and Spencer \cite{FSS}, Dyson, Lieb, and Simon \cite{DLS}, and   Fr\"ohlich, Israel, Lieb, and Simon \cite{FILS1} for more details.   In the context of nearest-neighbor hopping interactions, the vortex configuration of the ground state has been analyzed by Lieb \cite{Lieb1994} and Macris and Nachtergaele \cite{Nachtergaele1996}.  Recently one has shown that RP is also valid for a class of many-body Majorana interactions \cite{JP}; this  family of interactions includes the two-body $\H_{u}$ in \eqref{eq:u_Hamiltonian} with certain restrictions on the coupling constants $J_{ij}$.

\subsection{Vortex Loops and Reflection Positivity}
We study vortex loops $\widetilde{\mathfrak{B}}(\mathfrak{C})$ for ground states of the family of Hamiltonians $\{\H_{u}\}$ with ground state energies $\{\widetilde{E}_{0}(u)\}$.
 \begin{thm}\label{thm:main_theorem}
Let $\H_{u}$ denote a Hamiltonian of the form \eqref{eq:u_Hamiltonian}. Let $\mathfrak{C}$ denote a non-degenerate, reflection-symmetric loop with respect to a reflection $\vartheta$ in the plane $\Pi$. Assume that the magnitudes of the couplings are reflection-symmetric, $\abs{J_{(ij)}}=\abs{J_{(\vartheta i\, \vartheta j)}}$. Then $\min_{u}\widetilde{E}_{0}(u)$ is achieved for a ``vortex-free'' configuration of the $u_{ij}$'s, namely
	\be\label{eq:Theorem_Vortex_Contour}
	\widetilde{{\mathfrak{B}}}(\mathfrak{C})
	=1\,.
	\ee
\end{thm}
\begin{proof}
Consider a loop $\mathfrak{C}$ of length $2L$ symmetrically crossed by the hyper-plane $\Pi$. This means that $\Lambda_{\pm}\cap\mathfrak{C}$ each contain $L$ sites. Relabel the sites of $\mathfrak{C}$ as $1,\ldots,2L$ so that the bonds  in order on $\mathfrak{C}\cap\Lambda_{-}$ are $(i\,i+1)$ with  $i=1,\ldots,L-1$. Similarly  on  $\mathfrak{C}\cap\Lambda_{+}$ the bonds are $(i\,i+1)$ with $i=L+1,\ldots,2L-1$. 
Choose the starting point of $\mathfrak{C}$ so that the bonds cutting $\Pi$ are $(2L,\,1)$ and $(L,\,L+1)$. 

Define $\Lambda_{\Pi\,\pm}\subset\Lambda_{\pm}$ as those sites in $\Lambda_{\pm}$ that border $\Pi$.
Decompose $\H_{u}=\H_{u,-}+\H_{u,0}+\H_{u,+}$ where $\H_{u,\pm}\in\A_{\pm}$ and 
%	\be
%		\H_{u,0}=\sum_{i}J_{(i\,\vartheta i)}\,u_{i \,\vartheta i}\,\alpha_{\vartheta i}\,ic_{i}\vartheta(c_{i})\,,\qquad\text{with $i\in\Lambda_{\Pi\,-}$}\,.
%	\ee
	\be
		\H_{u,0}=\sum_{i}J_{(i\,\vartheta i)}\,u_{i \,\vartheta i}\,ic_{i}\vartheta(c_{i})\,,\qquad\text{with $i\in\Lambda_{\Pi\,-}$}\,.
	\ee
Perform a gauge transformation $W(\mathfrak{n})\in\mathfrak{G}^{c}$ of the form \eqref{eq:gauge_V}, with $n_{i}=0$ except for $i\in\Lambda_{\Pi\,-}$. 
Choose $n_{i}$  to ensure that the interactions in $\widetilde{\H}_{u,0}=W(\mathfrak{n})\,\H_{u,0}\,W(\mathfrak{n})^{*}$ across $\Pi$ are positive, namely
%	\be\label{eq:cond_new}
%	J_{(i\,\vartheta i)}\,u_{i\,\vartheta i}\,(-1)^{n_{i}}\,\alpha_{\vartheta i}>0\,,\qquad\text{for $i\in\Lambda_{\Pi\,-}$}\,.
%	\ee
\be\label{eq:cond_new}
	J_{(i\,\vartheta i)}\,u_{i\,\vartheta i}\,(-1)^{n_{i}}\,>0\,,\qquad\text{for $i\in\Lambda_{\Pi\,-}$}\,.
	\ee
Also define the Hamiltonians $\widetilde{\H}_{u,1}$ and $\widetilde{\H}_{u,2}$ as
	\be
		\widetilde{\H}_{u,1}
		=\widetilde{\H}_{u,-}+\widetilde{\H}_{u,0}+						\vartheta(\widetilde{\H}_{u,-})\,,
		\quad\text{and}\quad
		\widetilde{\H}_{u,2}
		=\vartheta(\widetilde{\H}_{u,+})+\widetilde{\H}_{u,0}+			\widetilde{\H}_{u,+}\,,
	\ee
where 
	\be
		\widetilde{\H}_{u,-}
		=W(\mathfrak{n})\,\H_{u,-}\,W(\mathfrak{n})^{*}\,,
		\qquad\text{and}\qquad
		\widetilde{\H}_{u,+}
		=W(\mathfrak{n})\,\H_{u,+}\,W(\mathfrak{n})^{*}\,.
	\ee
Since $\vartheta(\widetilde{\H}_{u,0})=\widetilde{\H}_{u,0}$, the Hamiltonians $\widetilde{\H}_{u,1}$ and  $\widetilde{\H}_{u,2}$ are  reflection-symmetric,
	\be\label{eq:RPH}
		\vartheta (\widetilde{\H}_{u,1})
		=\widetilde{\H}_{u,1}\,,
		\qquad\text{and}\qquad
		\vartheta (\widetilde{\H}_{u,2})
		=\widetilde{\H}_{u,2}\,.
	\ee
Furthermore the coupling constants in $\widetilde{\H}_{u,0}$ that cross the reflection plane $\Pi$ are positive.

The Hamiltonians $\widetilde{\H}_{u,1}$ and $\widetilde{\H}_{u,2}$ satisfy the hypothesis of Theorem 3 in \cite{JP}. In that paper one studies reflection-positivity for a class of interacting Majorana systems including the present one satisfying \eqref{eq:cond_new} and \eqref{eq:RPH}. From this result one concludes the reflection-positivity conditions. For $B\in\mathfrak{A}_{\pm}^{\text{even}}$,
	\be\label{eq:RP_HU1_HU2_new}
		\Tr_{\HH_{c}}\lrp{B\,\vartheta(B)\,e^{-\widetilde{\H}_{u,1}}}\geqslant0\,,
		\quad\text{and}\quad
			\Tr_{\HH_{c}}\lrp{B\,\vartheta(B)\,e^{-\widetilde{\H}_{u,2}}}\geqslant0\,.
	\ee
A direct consequence of the reflection-positivity conditions \eqref{eq:RP_HU1_HU2_new} is the reflection-positivity bound for any $\beta\geqslant0$, 
	\be\label{eq:trace_bound_new}
		\Tr_{\HH_{c}}\, e^{-\beta\widetilde{\H}_{u}}
		\leqslant
		\lrp{ \Tr_{\HH_{c}}\,e^{-\beta\widetilde{\H}_{u,1}}}^{1/2}	\lrp{ \Tr_{\HH_{c}}\, e^{-\beta\widetilde{\H}_{u,2}}}^{1/2}\,.
	\ee
This bound is a special case of the reflection-positivity inequality for interacting Majorana systems proved in Proposition 8 of \cite{JP}. 
The reflection-positivity bound \eqref{eq:trace_bound_new} allows one to establish an inequality on the ground state energy $\widetilde{\widetilde{E}}_{0}(u)$ of the Hamiltonian $\widetilde{\H}_{u}$ in terms of the ground-state energies $\widetilde{\widetilde{E}}_{0}(u,1)$ and $\widetilde{\widetilde{E}}_{0}(u,2)$ of the Hamiltonians $\widetilde{\H}_{u,1}$ and $\widetilde{\H}_{u,2}$, namely
		\be\label{eq:inequality_energies_new}
			0\geqslant\widetilde{\widetilde{E}}_{0}(u)\geqslant 						\frac{\widetilde{\widetilde{E}}_{0}(u,1)+\widetilde{\widetilde{E}}_{0}(u,2)}{2}\,.
		\ee
Taking $\beta$ large in \eqref{eq:trace_bound_new} proves \eqref{eq:inequality_energies_new}.

Conjugation by the gauge transformation $W(\mathfrak{n})$ does not change the ground state energy $\widetilde{E}_{0}(u)$ of $\H_{u}$, so $\widetilde{\widetilde{E}}_{0}(u)=\widetilde{E}_{0}(u)$. Nor does conjugation by the gauge transformation $W(\mathfrak{n})$ change the value of any vortex loop $\widetilde{B}(\mathfrak{C})$. Thus $\min_{u}\widetilde{E}_{0}(u)$ is obtained from some configuration $u=u_{0}$ that is both reflection-symmetric and has positive interactions across $\Pi$. Call this Hamiltonian $\H_{u_{0}}$.

Let $\H_{u_{0}}(\mathfrak{C})$ denote the Hamiltonian that is the restriction of $\H_{u_{0}}$ to bonds $(ij)\in\mathfrak{C}$.
Decompose $\H_{u_{0}}(\mathfrak{C})$ as
	\be
		\H_{u_{0}}(\mathfrak{C})=\H_{u_{0},-}(\mathfrak{C})+\H_{u_{0},0}(\mathfrak{C})+\H_{u_{0},+}(\mathfrak{C})\,,
	\ee 
where positivity of the $J_{(ij)}$'s ensures
	\beq
			\H_{u_{0},-}(\mathfrak{C})
			&=&\sum_{i=1}^{L-1}\,J_{(i,\, i+1)}\,u_{i\, i+1}\,i\,c_{i}c_{i+1}\,,\nn
			\H_{u_{0},+}(\mathfrak{C})
			&=&\sum_{i=L+1}^{2L-1}J_{(i,\,i+1)}\,u_{i\, i+1}\,i\,c_{i}c_{i+1}\,,\\
%		    \H_{u_{0},0}(\mathfrak{C})&=&J_{(1,\,2L)}\,u_{1\,2L}\,\alpha_{\vartheta 1}\,i\,c_{1}\,\vartheta(c_{1})+J_{(L,\,L+1)}\,u_{L\,L+1}\,\alpha_{\vartheta L}\,i\,c_{L}\,\vartheta(c_{L})\,.\nonumber
		    \H_{u_{0},0}(\mathfrak{C})&=&J_{(1,\,2L)}\,u_{1\,2L}\,i\,c_{1}\,\vartheta(c_{1})+J_{(L,\,L+1)}\,u_{L\,L+1}\,i\,c_{L}\,\vartheta(c_{L})\,.\nonumber
	\eeq 
With our chosen representation
%	\be\label{eq:cond_general_C_1_third_case}
%		J_{(1,\,2L)}\,u_{1\,2L}\,\alpha_{\vartheta 1}>0\,,\qquad\text{and}\qquad J_{(L,\,L+1)}\,u_{L\,L+1}\,\alpha_{\vartheta L}\,>0\,,
%	\ee
	\be\label{eq:cond_general_C_1_third_case}
		J_{(1,\,2L)}\,u_{1\,2L}>0\,,\qquad\text{and}\qquad J_{(L,\,L+1)}\,u_{L\,L+1}>0\,,
	\ee
and also reflection-symmetry $\vartheta(\H_{u_{0},-}(\mathfrak{C}))=\H_{u_{0},+}(\mathfrak{C})$ yields for $i=1,\ldots,L-1$,
%	\beq
%&&\hskip -2.0cm J_{(i,\, i+1)}\,u_{i\, i+1}\,\alpha_{2L-i}\,\alpha_{2L-i+1}\,i\,c_{2L-i}\,c_{2L-i+1}\nn
% 							&=&J_{(2L-i,\,2L-i+1)}\,u_{2L-i\ 2L-i+1}\ i\,c_{2L-i}\,c_{2L-i+1}\,.
%		\eeq
	\beq
&&\hskip -2.0cm J_{(i,\, i+1)}\,u_{i\, i+1},i\,c_{2L-i}\,c_{2L-i+1}\nn
 							&=&J_{(2L-i,\,2L-i+1)}\,u_{2L-i\ 2L-i+1}\ i\,c_{2L-i}\,c_{2L-i+1}\,.
		\eeq
Consequently, since the loop $\mathfrak{C}$ is non-degenerate, for $i=1,\ldots,L-1$ one has
%	\be	\label{eq:cond_general_C_2_third_case}
%		J_{(i,\,i+1)}J_{(2L-i,\,2L-i+1)}\,u_{i\, i+1}\,u_{2L-i\ 2L-i+1}\ \alpha_{2L-i}\,\alpha_{2L-i+1}>0\,.
%	\ee  
	\be	\label{eq:cond_general_C_2_third_case}
		J_{(i,\,i+1)}J_{(2L-i,\,2L-i+1)}\,u_{i\, i+1}\,u_{2L-i\ 2L-i+1}>0\,.
	\ee  

Multiply together conditions \eqref{eq:cond_general_C_1_third_case} with all the conditions \eqref{eq:cond_general_C_2_third_case}, and identify site $2L+1$ with site $1$. One obtains
	\be\label{eq:result_third_case}
		\mathfrak{B}(\mathfrak{C})
		=-\prod_{(ij)\in\mathfrak{C}}\,u_{ij}
		=-\prod_{i=1}^{2L}\, u_{i\, i+1}
		=\text{sgn}\lrp{\prod_{i=1}^{2L}\,J_{(i,\,i+1)}}
		=1\,.
	\ee
The first two equalities and the last equality in \eqref{eq:result_third_case} are definitions, so one only needs to verify the third equality. There is one additional minus sign, which comes from $u_{1\,2L}=-u_{2L\,1}$, with the former appearing in \eqref{eq:cond_general_C_1_third_case} and the later in the product $\prod_{i=1}^{2L}\, u_{i\, i+1}$. This minus sign cancels the explicit minus sign in \eqref{eq:result_third_case}.
$\hfill \qed$
\end{proof}

\setcounter{equation}{0}
\section{Quantum Spin Ladders}\label{sec:Quantum_Spin_Ladders}
One way to realize the family of Hamiltonians $\H_{u}$ defined in \eqref{eq:u_Hamiltonian} is to study nearest-neighbor spin interactions on a trivalent lattice. We consider the simplest example, the quantum spin ladder, corresponding to the case $d=2$ in \S\ref{sec:Majorana_Cubic_Lattice}.

\subsection{Even Spin Ladders}\label{sec:even_ladder}
An  open, even spin ladder  is a $2\times 2N$ square lattice array. The sites of the lattice are connected by bonds linking nearest-neighbor sites. 
We call one given plaquette the {\em unit cell} of the ladder. One  obtains the lattice of the ladder as a union of  $N$ translates of the unit cell by integer multiples of twice the side-length of the unit cell, along one of its coordinate axes (which we choose horizontal).   One completes the ladder with bonds  $(ij)$ that link site $i$ with a nearest-neighbor site $j$.

\begin{figure}[h]
	\centering
		\includegraphics[width=0.9\textwidth]{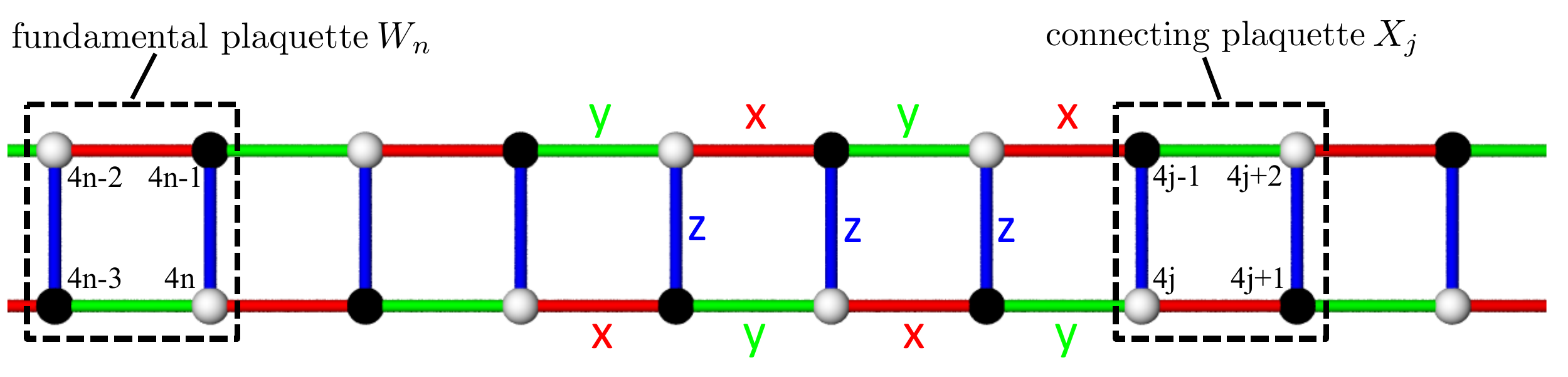}
	\caption{Ladder}
	\label{fig:Ladder}
\end{figure}

We illustrate such a ladder in Fig.~\ref{fig:Ladder}, where we label the plaquettes, vertices, and bonds.
Divide the $2N-1$ plaquettes of the ladder into two sets: the first set comprises $N$ {\em fundamental plaquettes} $p_{1}, p_{3},\ldots,p_{2k-1}, \ldots, p_{2N-1}$ that are the translates of the unit cell that generates the ladder. The other set contains $(N-1)$ {\em connecting plaquettes} $p_{2},p_{4}, \ldots,p_{2k},\ldots, p_{2N-2}$, each of which links two fundamental plaquettes, by sharing two of its  bonds with two different fundamental plaquettes.  

In order to discuss both ``open'' and ``closed'' ladders in a unified way, we introduce one additional connecting plaquette $p_{2N}$ linking $p_{2N-1}$ with $p_{1}$, and two additional bonds to the open ladder, connecting the site $4N$ to the site $1$, and connecting the site $4N-1$ to the site $2$.  The closed ladder corresponds to periodic boundary conditions.
Another way to characterize a closed ladder, is the property that one must remove at least four bonds to divide it into two disconnected pieces.

Label the sites in the fundamental plaquette $p_{2n-1}$ clockwise, starting in the lower-left corner, by $4n-3, 4n-2, 4n-1, 4n$. As a consequence, the sites in the connecting plaquette $p_{2j}$ are labeled clockwise by $4j, 4j-1, 4j+2, 4j+1$.  
The open ladders we consider have $(6N-2)$ bonds, which we divide into three types.  There are $(2N-1)$ type-$x$ bonds, $(2N-1)$ type-$y$ bonds, and $2N$ type-$z$ bonds.  All the vertical bonds  will be type-$z$ bonds.  The horizontal bonds on top of each fundamental plaquette, and on the bottom of each connecting plaquette  are type-$x$ bonds.  The remaining bonds are type-$y$ bonds.

\subsubsection{The Hamiltonian}
 The models we study here and other similar models arise frequently in the study of topological quantum information theory, see for instance \cite{Feng2007,Pachos2008,Nussinov2008,Saket2010,Motrunich2011,Vishveshwara2011,Chesi2012,Chen2013,Wen2012}. 
The  spins at each site  $\vec\sigma_{i}=(\sigma_{i}^{x}, \sigma_{i}^{y}, \sigma_{i}^{z})$ are Pauli matrices.  Here $i$ denotes the lattice site (using the labels above), and $x,y,z$ denotes the three Pauli matrices.    
The Hamiltonian we study is a nearest-neighbor quadratic interaction of the form
	\begin{equation}
	\label{eq:Hamiltonian}
		H
		= -\sum_{(ij)}J_{(ij)}\,\sigma_{i}^{{(ij)}}\sigma_{j}^{{(ij)}}\;, \qquad J_{(ij)}=J_{(ji)}\,\text{ real.}
	\end{equation}
Here the sum over $(ij)$ denotes a sum over unoriented bonds $(ij)$ between nearest neighbor lattice sites in the ladder.  
Also $\sigma_{i}^{(ij)}$ equals $\sigma_{i}^{x}$, $\sigma_{i}^{y}$, or $\sigma_{i}^{z}$, according to whether the bond $(ij)$  is type-$x$, type-$y$, or type-$z$, respectively, as defined above; thus the couplings labeled by a bond depends only on products of the same components of $\vec \sigma$ at different sites. 

A simple case of this Hamiltonian which we call \emph{homogeneous couplings} is the case for which every type-$x$ bond has coupling $J_{x}$, every type-$y$ bond has coupling $J_{y}$, and every type-$z$ bond has coupling $J_{z}$. 
The open ladder Hamiltonian corresponds to taking the two coupling constants closing the ladder equal to zero, namely $J_{(4N\,1)}=J_{(4N-1\,2)}=0$.

\subsection{Vortex Loops}
For each loop $\mathfrak{C}$, we assign a   {\em vortex-loop operator} (or simply a {\em vortex}) $\mathfrak{B}(\mathfrak{C})$.  This is proportional to the {\em ordered} product along the loop of terms in the interaction. Recall that  $\sigma_{i}^{{(ij)}}\sigma_{j}^{{(ij)}}$ is the term in the Hamiltonian \eqref{eq:Hamiltonian} on the bond $(ij)$.  Define 
	\beq\label{eq:Vortex Loop Operator}
		\mathfrak{B}(\mathfrak{C})
		&=&i^{\abs{\mathfrak{C}}+2} \prod_{(ij)\in \mathfrak{C}}\lrp{ \sigma_{i}^{{(ij)}}\sigma_{j}^{{(ij)}}}\;,
	\eeq
similar to (5-6) in Kitaev \cite{KitaevHoney}. As $\sigma_{i}^{(ij)}$ is self-adjoint with square $I$, we infer that $\mathfrak{B}(\mathfrak{C})$ is unitary.  
We devote the rest of this paper to the study of properties of the operators $\mathfrak{B}(\mathfrak{C})$. 

\setcounter{equation}{0}
\section{Fermionic Ladders}\label{sec:Fermionic_Ladders}
\subsection{Mapping of Spins to Fermions}
We use a representation of the Pauli matrices as quadratic expressions in Majoranas. Choose four Majoranas  at each lattice site $j$ and denote them $b_{j}^{x}, b_{j}^{y}, b_{j}^z$, and $c_{j}$.
Define the  algebra $\mathfrak{A}$ as the algebra generated by the $b_{j}^{x,y,z}$ and $c_{j}$ for $j\in\Lambda$. Let $\A^{\text{even}}$ denote the even subalgebra of $\mathfrak{A}$, generated by even monomials in the fermionic operators.
One defines for the single site $j$,
	\begin{equation}
		\label{eq:Spin Representation}
		\widetilde{\sigma}_{j}^{\alpha}
		=ib_{j}^{\alpha}\,c_{j}\;,
	\end{equation}
as in the usual construction of boost operators for the Dirac equation, see for example \S 4c of \cite{Schweber}. For a single chiral component of the spin-$1/2$ Dirac wave function, the boost generator is isomorphic to the spin.

Denote the vector space of the four Majoranas as  $\widetilde{\mathcal H}_{j}$. In order to project onto a single chiral component, one restricts to the eigenspace $+1$ of the self-adjoint matrix $\gamma_{j}^{5}=b_{j}^{x}\, b_{j}^{y}\, b_{j}^z\, c_{j}$ with square one. On the full Hilbert space $\H$ these $\gamma_{j}^{5}$'s mutually commute and commute with each $\widetilde{\sigma}^{x,y,z}_{j}$. 
 The corresponding orthogonal projection onto this eigenspace is $P=\prod_{j}P_{j}=\prod_{j}\frac{1}{2}\left(I+\gamma_{j}^{5}\right)$, and it yields $\mathcal H=P\widetilde{\mathcal H}$ of dimension $2^{4N}$.
The $\sigma_{j}^{\alpha}=P_{j} \,\widetilde \sigma_{j}^{\alpha}\, P_{j}$'s satisfy the correct multiplication laws for spin matrices on $\mathcal{H}$.
%In the case of the ladder, the tensor product $\widetilde {\mathcal H}=\otimes_{j=1}^{4N} \widetilde {\mathcal H}_{j}$ has dimension $2^{8N}$, and we call $\widetilde {\mathcal H}$ the {\em fermionic Hilbert space}. As  the matrices $\gamma_{j}^{5}$ commute at different sites $j$, the $4N$ projections $P_{j}$ mutually commute on $\widetilde {\mathcal H}$.  The projections $P_{j}$ also commute with all the spins $\vec{\widetilde \sigma}_{i}$.  Thus the projection
%	\be\label{eq:Physical_Projection}
%			P=
%			\prod_{j=1}^{4N} P_{j}
%	\ee		
%has a $2^{4N}$ dimensional range $\mathcal H=P\widetilde{\mathcal H}$, which is the correct dimension of the tensor product space for $4N$ Pauli spins.  Furthermore the representation of the field
% 	\begin{equation}
%		\vec\sigma_{j}
%		= P \,\vec{\widetilde \sigma_{j}}\, P
%		= P \,i\vec {b}_{j}\,c_{j}\, P\;,
%	\quad \text{acts on } \ \ {\mathcal H}\;.
%	\label{Multiple Spin Representation}
%	\end{equation}
We call $\mathcal H$ the {\em spin-ladder Hilbert space}.

\subsection{Representation of the  Hamiltonian}
Introduce the three skew $4N\times 4N$ matrices $ u$, $A$, and $C$  with entries that are hermitian operators,
	\be
	u_{ij}=-u_{ji}=u_{ij}^{*}\,,
	\qquad 
	A_{ij}=-A_{ji}=A_{ij}^{*}\,,
	\qquad\text{and}\quad
	C_{ij}=-C_{ji}=C_{ij}^{*}\,.
	\ee
We define these matrix elements to vanish unless $i,j$ are nearest-neighbors. In this case
	\be
		{u}_{ij}
		=
		 i b_{i}^{(ij)}\,b_{j}^{(ij)}\,,
		 \qquad
		 C_{ij}
		 =
		 ic_{i}c_{j}\,,
		 \qquad
		 \text{and}
		 \quad
		{A}_{{ij}}
		=J_{{(ij)}}\,{u}_{{ij}}\,,
	\label{eq:u and a}
	\ee
with $J_{\lrp{ij}}=J_{\lrp{ji}}$ real.
A representation of the spin-ladder Hamiltonian on the fermionic Hilbert space is 
	 \begin{eqnarray}\label{eq:HamiltonianExtended}
		\widetilde{H}
		=\sum_{\lrp{ij}}{A}_{{ij}}\,C_{ij}\;
		=\sum_{\lrp{ij}} J_{(ij)}  \,u_{ij}   \,i\,c_{i}\,c_{j}
		=\widetilde{H}^{*}\;.
	\end{eqnarray}

The $u_{{ij}}$ operators mutually commute, and they also commute with the Hamiltonian $\widetilde{H}$.  They satisfy $u_{{ij}}^2=+1$, so the eigenvalues of $u_{ij}$ are $\pm1$. Also all the  $\gamma_{j}^5$ commute with $\widetilde{H}$.
Furthermore the Hamiltonian $\widetilde{H}$ commutes with $P$, so it maps the subspace $\mathcal H$ into itself and on this subspace the Hamiltonian has the representation as a sum of self-adjoint operators,
\begin{equation}
		H 
		= P\widetilde H P
		= \sum_{\lrp{ij}}  P\,A_{ij}\,C_{ij}\,P\;.
	\end{equation}
The properties of $\widetilde{H}$ on $\widetilde{\mathcal{H}}$ are different from those of $H$ on $\mathcal{H}$, and in particular the eigenvalues might different (aside from multiplicity). One should be careful not to jump to conclusions, as we give numerical evidence for the existence of both types of behavior in \S\ref{sec:Numerical_Evidence}, see also \cite{PCL}.

\subsection{Representation of the Vortices}
A fermionic representation of $\widetilde{\mathfrak{B}}(\mathfrak{C})$ of $\mathfrak{B}(\mathfrak{C})$ commutes with the projection $P$.  Its projection       
	$
		P\, \widetilde {\mathfrak{B}}(\mathfrak{C}) \, P
	$,
agrees with the original definition \eqref{eq:Vortex Loop Operator} of the vortex $\mathfrak{B}(\mathfrak{C})$. 
We give such a fermionic representation $\widetilde{\mathfrak{B}}(\mathfrak{C})$, similar to \cite{KitaevHoney} and observe that the spin vortices $\mathfrak{B}(\mathfrak{C})$ are mutually commuting, conserved quantities. 

\begin{prop}\label{prop:productu}
A fermionic representation of the vortex loop-operator is given in terms of the mutually-commuting operators $u_{ij}$ as  
	\be\label{eq:vortex_operators_fermions}
		\widetilde {\mathfrak{B}}(\mathfrak{C})
		=- \prod_{\lrp{ij}\in\mathfrak{C}} u_{ij}\,.
	\ee
Each $\gamma_{k}^{5}$ commutes with $\widetilde{\mathfrak{B}}(\mathfrak{C})$, namely
	\be
		\comm{\widetilde{\mathfrak{B}}(\mathfrak{C})}{\gamma_{k}^{5}}=0\,.
	\ee

\end{prop}

\begin{proof}
The contribution to the vortex loop-operator ${\mathfrak{B}}(\mathfrak{C})$ from the spins at site $i_{j}$, for $j\neq1$, is $\sigma_{i_{j}}^{(i_{j-1}\,i_{j})}\sigma_{i_{j}}^{(i_{j}\,i_{j+1})}$.  (In case $j=\ell$, set $\ell+1=1$.)  This product has the fermionic representation
$-b_{i_{j}}^{(i_{j-1}\,i_{j})}\,c_{j}\,b_{i_{j}}^{(i_{j}\,i_{j+1})}\,c_{j}=b_{i_{j}}^{(i_{j-1}\,i_{j})}\,b_{i_{j}}^{(i_{j}\,i_{j+1})}$.  Taking the product of these representations and adding the contribution from the spins at site $i_{1}$, 
one has a fermionic representation for $ \mathfrak{B}(\mathfrak{C}) $ defined in  \eqref{eq:Vortex Loop Operator}  equal to 
	\beq
		\widetilde {\mathfrak{B}}(\mathfrak{C}) 
		 &=& - i^{\abs{\mathfrak{C}}+2}\,b_{i_{1}} ^{(i_{1}\,i_{2})}\,c_{i_{1}}\,
		 b_{i_{2}}^{(i_{1}\,i_{2})}\,b_{i_{2}}^{(i_{2}\,i_{3})}\,
		 b_{i_{3}}^{(i_{2}\,i_{3})}\,b_{i_{3}}^{(i_{3}\,i_{4})}\,
		 \cdots\nn
		 &&\qquad \,b_{i_{\ell-1}}^{(i_{\ell-2}\,i_{\ell-1})} \,
		 b_{i_{\ell-1}}^{(i_{\ell-1}\,i_{\ell})}\,
		 \,b_{i_{\ell}}^{(i_{\ell-1}\,i_{\ell})} \,
		 b_{i_{\ell}}^{(i_{\ell}\,i_{1})}\,
		 b_{i_{1}}^{(i_{\ell}\,i_{1})} 
		 \,c_{i_{1}}\nn
		 &=& i^{\abs{\mathfrak{C}}+2}\,c_{i_{1}}\, b_{i_{1}} ^{(i_{1}\,i_{2})}\,
		 b_{i_{2}}^{(i_{1}\,i_{2})}\,b_{i_{2}}^{(i_{2}\,i_{3})}\,
		 b_{i_{3}}^{(i_{2}\,i_{3})}\,b_{i_{3}}^{(i_{3}\,i_{4})}\,
		 \cdots\nn
		 &&\qquad \,b_{i_{\ell-1}}^{(i_{\ell-2}\,i_{\ell-1})} \,
		 b_{i_{\ell-1}}^{(i_{\ell-1}\,i_{\ell})}\,
		 \,b_{i_{\ell}}^{(i_{\ell-1}\,i_{\ell})} \,
		 b_{i_{\ell}}^{(i_{\ell}\,i_{1})}\,
		 b_{i_{1}}^{(i_{\ell}\,i_{1})} 
		 \,c_{i_{1}}\nn
		 &=& - c_{i_{1}} \, u_{i_{1}i_{2}}\,u_{i_{2}i_{3}}\, \cdots  \,u_{i_{\ell-1}i_{\ell}}\,u_{i_{\ell}i_{1}}\,c_{i_{1}}\nn
		 &=& -  u_{i_{1}i_{2}}\,u_{i_{2}i_{3}}\, \cdots  \,u_{i_{\ell-1}i_{\ell}}\,u_{i_{\ell}i_{1}}\;.
	\eeq
In the last equality we use the fact that $c_{i_{1}}$ commutes with all the $u_{ij}$'s. This establishes the fermionic representation \eqref{eq:vortex_operators_fermions}.  As each $u_{ij}$ is hermitian and the $u_{ij}$ mutually commute, we infer that $\widetilde {\mathfrak{B}}(\mathfrak{C}) $ is hermitian.
Since $\widetilde {\mathfrak{B}}(\mathfrak{C}) $ is a product of $b$ Majoranas, with an even number of $b$'s at each site $i_{j}\in\mathfrak{C}$,  we infer that $\widetilde {\mathfrak{B}}(\mathfrak{C}) $ commutes with each $\gamma_{j}^{5}$.  Therefore $\widetilde {\mathfrak{B}}(\mathfrak{C}) $ commutes with $P$.     
$\hfill\qed$
\end{proof}

From the representation \eqref{eq:vortex_operators_fermions} for $\widetilde {\mathfrak{B}}(\mathfrak{C})$ and the representation \eqref{eq:HamiltonianExtended} for $\H$ in terms of the mutually-commuting, self-adjoint operators $u_{ij}$ with square one, one infers the following two corollaries: 

\begin{cor}
The fermionic vortex representatives $\widetilde{\mathfrak{B}}(\mathfrak{C})$ are all self-adjoint and have eigenvalues $\pm1$. 
Different $\widetilde{\mathfrak{B}}(\mathfrak{C})$  mutually commute,
	\be
		\comm{\widetilde{\mathfrak{B}}(\mathfrak{C})}{\widetilde{\mathfrak{B}}(\mathfrak{C}')}
		=0\,.
	\ee
All the $\widetilde{\mathfrak{B}}(\mathfrak{C})$  are conserved by $\H$, namely
	\be
		\comm{\widetilde{\mathfrak{B}}(\mathfrak{C})}{\H}
		=0\;.
	\ee
\end{cor}

\begin{cor}
The vortex loop operators ${\mathfrak{B}}(\mathfrak{C})$ are self-adjoint  on $\mathcal{H}$, and have eigenvalues $\pm1$.  Different  ${\mathfrak{B}}(\mathfrak{C})$ mutually commute,
	\be
		\comm{{\mathfrak{B}}(\mathfrak{C})}{{\mathfrak{B}}(\mathfrak{C}')}
		=0\,.
	\ee
The vortex loop-operators are all conserved, namely
	\be
		\comm{{\mathfrak{B}}(\mathfrak{C})}{H}
		=0\;.
	\ee
\end{cor}

\subsection{The Reduced Fermionic Hamiltonians}  
Define $\H_{u}$ as the Hamiltonian $\H$ restricted to an eigenspace of the $u_{ij}$'s. Therefore it is useful to represent the Hilbert space $\HH$ in the form of a tensor product 
	\be\label{eq:tensor_product_of_H}
		\HH=\HH_{u}\otimes\HH_{c}\,.
	\ee
Here we consider the $6N$ mutually commuting variables $u_{ij}$ corresponding to the products of $ib_{i}^{(ij)}b_{j}^{(ij)}$ on the $6N$ bonds of a closed ladder. In the case of an open ladder the couplings on the two extra bonds $(1,\, 4N)$ and $(2,\, 4N-1)$ are zero. Each $u_{ij}$ is self-adjoint and has square equal to one, so it can be represented on a two-dimensional Hilbert space. Therefore the Hilbert space $\H_{u}$ has dimension $2^{6N}$, which is exactly $2^{\#_{b}/2}$, where $\#_{b}$ equals the total number of $b^{x,y,z}$ Majoranas. These Majoranas can be represented on a Hilbert space of the same dimension $2^{6N}$.

Define the fermionic algebra $\mathfrak{A}_{c}$ as the subalgebra of $\A$ generated by the $c_{j}$-Majoranas. Since this algebra commutes with all the $u_{ij}$'s, it acts as $I\otimes \mathfrak{A}_{c}$ on $\HH=\HH_{u}\otimes\HH_{c}$.

\setcounter{equation}{0}
\section{Eigenvalues of $\widetilde H$ and of $H$}\label{sec:Exact_Results}
Let $\widetilde E_{0}$ denote the ground-state energy of $\widetilde H$ given in \eqref{eq:HamiltonianExtended}, and let $E_{0}$ denote  the ground-state energy of $H$. We are interested to know when these two ground state energies coincide. 
By the variational principle, there is a normalized vector $\widetilde \Omega\in\calfh$, such that 
	\[
	\widetilde E_{0}
	=\lra{\widetilde\Omega, \widetilde H \widetilde \Omega} 
	= \inf_{\Vert \widetilde \chi\Vert=1}  
	\lra{\widetilde\chi, \widetilde H \widetilde \chi} 
	\le 
	E_{0}\;.
	\] 
One obtains $E_{0}$ by restricting $\widetilde \chi$ to the range of $P$. So if $P\widetilde\Omega=\widetilde\Omega$, then $\widetilde E_{0}=E_{0}$.  More generally, we investigate the eigenvalues of $\widetilde H$, and determine in certain cases  that they are the same as the eigenvalues of  $H$.  In other cases there is evidence that they are different.

For an open ladder, we prove in Theorem \ref{prop:multiplicity} that $\widetilde{H}$ and $H$ have the same eigenvalues. We analyze the ground state of $H$ using the fermionic representation and demonstrate that the ground state is vortex-free.  In reflection symmetric cases we do this in \S\ref{sec:Ref_symm_case}  using reflection positivity.  In \S\ref{sec:non reflection-symm} we analyze some non-reflection symmetric cases using perturbation theory. 

In Proposition \ref{prop:monomials} we explain why the proof of Theorem \ref{prop:multiplicity}  for the open ladder does not apply to the closed ladder. More to the point, numerical calculation shows that the spectra are really different, see the discussion in \S\ref{sec:Numerical_Evidence} and in particular in \S\ref{sec:Numerical Closed Ladders}.

\begin{thm}\label{prop:multiplicity}
Consider an open ladder. The eigenvalues of $H$ defined in \eqref{eq:Hamiltonian} are the same as those of  $\widetilde{H}$ defined in \eqref{eq:HamiltonianExtended}, aside from multiplicity. 
\end{thm}

\begin{proof}
The operators $\gamma_{i}^{5}$ mutually commute and commute with $\widetilde{H}$, so we can simultaneously diagonalize them. We find an operator $Q_j$ with square $\pm1$, which anti-commutes with $\gamma_{j}^{5}$ and commutes with $\widetilde{H}$ and $\gamma_{i}^5$, for $i\neq j$. Let $\widetilde{\Omega}$ be a simultaneous eigenstate of the $\gamma_{i}^{5}$ and $\fh$ with eigenvalues ($\mu_{i}$, $\widetilde{E}$), where $\mu_i=\pm1$ are the eigenvalues of the $\gamma_{i}^{5}$. Then the vector $Q_j\widetilde{\Omega}$ is an eigenstate with the same eigenvalues except for the one $\mu_i$ with $i=j$, that has the opposite sign. (Note $Q_{j}\widetilde{\Omega}\neq 0$, as $Q_j^2=\pm1$.)  By applying $Q_{j}$ for each negative $\mu_{j}$, we obtain   a simultaneous eigenstate  with energy $\widetilde{E}$, and with all the $\mu_i=+1$.  Calling this vector $ {\widetilde \Omega}'$, the projected state $P\widetilde{\Omega}'=\Omega'$ is an eigenstate of $H$ with eigenvalue $\widetilde{E}$.  This also shows that to each eigenvalue $E$ of $H$ is associated $2^{4N}$ eigenvalues of $\fh$, of which all but one of the corresponding eigenvectors project to zero.  

Define the operator $Q_j$  by considering a non-self-intersecting path $\Gamma$ through the ladder from site $j$ to site $4N$. The operator $Q_j$ equals the product of the $u_{i'j'}$ operators along the bonds $\lrp{i'j'}$ on this path, followed by $b_{4N}^x$. This $Q_{j}$ is a product of $b$ operators, so its square is $\pm1$. The operator $b_{4N}^{x}$ does not enter the expression  \eqref{eq:HamiltonianExtended} for $\widetilde{H}$, and each term in $\widetilde H$ is a product of an even number of other fermion operators. Therefore  $Q_j$ commutes with $\widetilde{H}$. 

To complete the proof, we need to check the commutativity of $Q_{j}$ with the operators  $\gamma_{i}^{5}$.  Consider four cases: first suppose the path $\Gamma$ does not pass through  $i$.  Then $\gamma_{i}^{5}$ commutes with each $b$ belonging to $Q_{j}$, so it commutes with $Q_{j}$.   

Second suppose that $i$ is a site on the path $\Gamma$, but $i\neq j$ and $i\neq 4N$.  In this case the site $i$ contributes a product of two different $b_{i}$ operators to $Q_j$; this is the case, because in the ladders we consider, the three bonds ending at site $i$ are of three different types, and the path $\Gamma$ contains two of these bonds.  Each of these two $b_{i}$'s anti-commutes with $\gamma_{i}^{5}$, so their product commutes.  Also $\gamma_{i}^{5}$ commutes with $b_{k}$'s at other sites, so it commutes with $Q_{j}$.   

The third case is $i=4N$.  As before, $\gamma_{4N}^{5}$ commutes with the $b$'s at sites different from $4N$. Only one bond in $\Gamma$ ends at site $4N$, so only one $b_{4N}$ at site $4N$ arises from the path; for our ladders, this must be either $b^{y}_{4N}$ or $b^{z}_{4N}$.  But $Q_{j}$ also includes the extra $b^{x}_{4N}$. So $\gamma_{4N}^{5}$ anti-commutes with this extra $b_{4N}^{x}$ and therefore commutes with the product of the two distinct $b_{4N}$'s that occur in $Q_{j}$.  

The fourth case is $i=j$. In this case only one bond in $\Gamma$ enters site $i$, so only one $b_{i}$ occurs in $Q_{j}$.  Hence $\gamma_{j}^{5}$ anti-commutes with the $b_{j}$'s in $Q_{j}$.  As $\gamma_{j}^{5}$ commutes with the $b$'s at other sites, $\gamma_{j}^{5}$ anti-commutes with $Q_{j}$. These cases cover all possibilities, so we have established all the desired properties of the operators $Q_{j}$. 
$\hfill\qed$
\end{proof}

We remark that an alternate proof could be based on the explicit form of the projection $P:\HH\to \mathcal{H}$ as a function of the variables $u_{ij}$ derived in Appendix A of  \cite{PCL}.
We now show that the proof of Theorem \ref{prop:multiplicity} does not extend in a straightforward way to the closed ladder. This is in line with the numerical calculations we perform in \S\ref{sec:Numerical_Evidence} suggesting that the spectrum of $H$ is different from the spectrum of $\tilde{H}$ for the closed ladder.

\begin{prop}\label{prop:monomials}
Consider a closed ladder Hamiltonian $\H$ of the form \eqref{eq:HamiltonianExtended} with all couplings $J_{(ij)}$ different from 0. There is no non-zero monomial $Q_j$ in the $b$'s and $c$'s that anti-commutes with $\gamma_{j}^{5}$ and commutes with $\widetilde{H}$ and $\gamma_{k}^{5}$ for $k\neq j$.
\end{prop}
\begin{proof}
Each site $k$ in the ladder gives rise to a $4$-dimensional  Hilbert space $\HH_{k}$.  There are $16$ linearly-independent operators on $\HH_{k}$, and this space is spanned by monomials $M_{k}^{\alpha}$ in the $b_{k}^{x,y,z}$ or $c_{k}$ of degree 4 or less. Of these, four monomials that we denote $m_{k}^{1,\ldots,4}$ are the Majoranas themselves and have degree $1$, and four others $m_{k}^{1,\ldots,4}\gamma_{k}^{5}$ have degree $3$.  We write these eight odd degree monomials as $M_{k}^{-,\alpha}$.  Each $M_{j}^{-,\alpha}$ anti-commutes with $\gamma_{j}^5$ and commutes with $\gamma_{k}^{5}$ for $k\neq j$. 

There are eight monomials $M_{k}^{+,\alpha}$ of degree $0$, $2$, or $4$, and these commute with all the  $\gamma_{j}^{5}$. All $16$ of the  $M_{k}^{\pm, \alpha}$  commute with $\gamma_{k'}^{5}$ for $k'\neq k$. The monomials in the $b$'s and $c$'s are linearly independent and span the operators on $\HH$, as shown in Proposition 1 of \cite{JP}. From these properties, we infer  that 
	\be\label{Qj}
	Q_{j}=\pm\, M_{j}^{-,\alpha}\prod_{k\neq j}M_{k}^{+,\alpha_k}\;.
	\ee

We now consider further restrictions  on $Q_{j}$, imposed by the fact that  one wants $\comm{Q_{j}}{\H}=0$.  We show this is impossible for $Q_{j}$ of form \eqref{Qj}.  These restrictions use the assumption that all  $J_{(ij)}\neq0$, so they do not apply in the case of an open ladder.   

Let us denote the interaction on bond $(ji)$ by  $\lra{ji}$, so the Hamiltonian \eqref{eq:HamiltonianExtended} can be written 
	\be\label{HNewNotation}
		\H
		= \sum_{(ji)} \lra{ji}\;,
		\qquad\text{where}\quad
		\lra{ji} = J_{(ji)}\,u_{ji}\,ic_{j}c_{i}
		 =- J_{(ji)}\,b_{j}^{(ji)}\,b_{i}^{(ji)}\,c_{j}c_{i}\;.
	\ee
We claim that 
\begin{enumerate}
\item[I.] $M_{j}^{-,\alpha}$ anti-commutes with either one or three terms in the sum \eqref{HNewNotation}.  

\item[II.]  $\prod_{k\neq j}^{3} M_{k}^{+,\alpha_{k}}$ anti-commutes with an even number of terms in \eqref{HNewNotation}.
\end{enumerate}
These two properties show that $Q_{j}$ of the form \eqref{Qj} cannot commute with $\H$.

In order to establish property (I), notice that a single Majorana  $c_{j}$ anti-commutes with three terms $\lra{ji}$ in the sum \eqref{HNewNotation}, where $i$ are the  three nearest neighbors to $j$.  Also the Majorana $b_{j}^{x,y,z}$ anti-commutes with one such term. As $\gamma_{j}^{5}$ commutes with $\lra{ji}$, the same anti-commutativity properties hold for $m_{j}^{1,\ldots,4}\gamma_{j}^{5}$ as for $m_{j}^{1,\ldots,4}$.

Property (II) also follows by considering the anti-commutation properties of the eight possible $M_{k}^{+,\alpha_{k}}$.  The identity and monomial of degree $4$ commute with each $\lra{ji}$. The monomials $M_{k}^{+,\alpha_{k}}$ of degree $2$ all anti-commmute with two of the $\lra{ji}$'s.   The statement then follows. 
$\hfill\qed$
\end{proof}

\section{Ladder Hamiltonians and Reflections \label{sect:Ladders and Reflections}}
In the following we consider ladder Hamiltonians $\H$ of the form \eqref{eq:HamiltonianExtended}  with reflection-symmetric absolute value of the couplings $J_{(ij)}$, namely
	\be\label{eq:ref_sym_couplings}
		\abs{J_{(\vartheta i\, \vartheta j)}}=\abs{J_{(ij)}}\,.
	\ee
We determine the value of reflection-symmetric vortex loops in the ground states of $\H$ and $H$ for such couplings.

The open or closed ladder in Fig.~\ref{fig:Ladder} satisfies \eqref{eq:ref_sym_couplings} in three cases:

\smallskip
\noindent{\bf Case I. Reflection through a horizontal plane, see Fig.~\ref{fig:LadderI}}.
We make no restriction on the couplings $J_{(i\,i+1)}$ on vertical bonds.

\noindent{{\bf Case II. Vertical reflection plane bisecting an open ladder, see Fig.~\ref{fig:LadderII}.}

\noindent{{\bf Case III. Reflection through any vertical plane bisecting a closed ladder, see Fig.~\ref{fig:LadderIII}.} The dotted reflection plane intersects the ladder twice.

\begin{figure}[h!]
	\centering
		\includegraphics[width=0.9\textwidth]{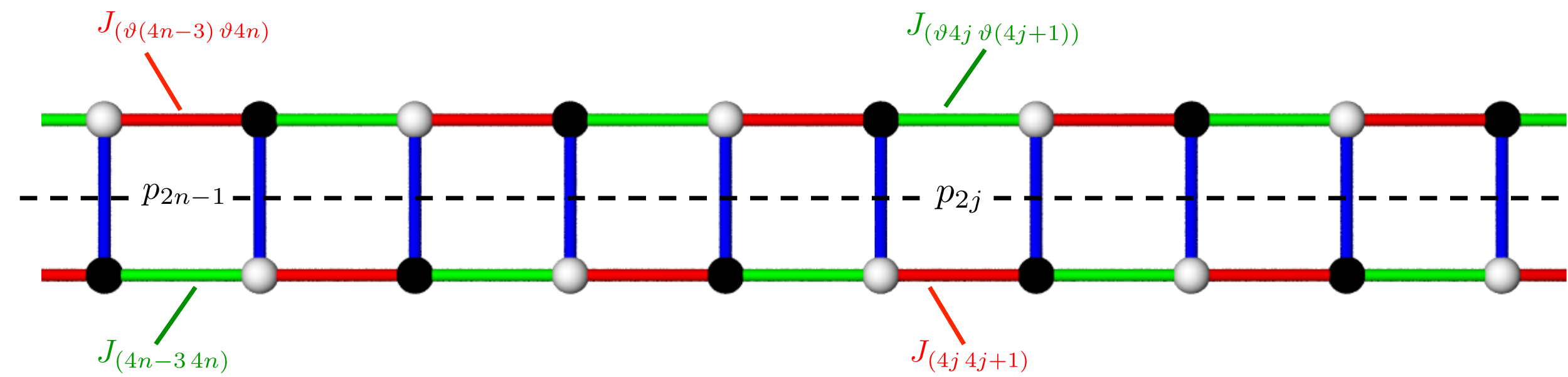}
	\caption{Case I: Horizontal reflection plane for an open or closed ladder.}
	\label{fig:LadderI}
\end{figure}

	\begin{figure}[h!]
	\centering
		\includegraphics[width=0.9\textwidth]{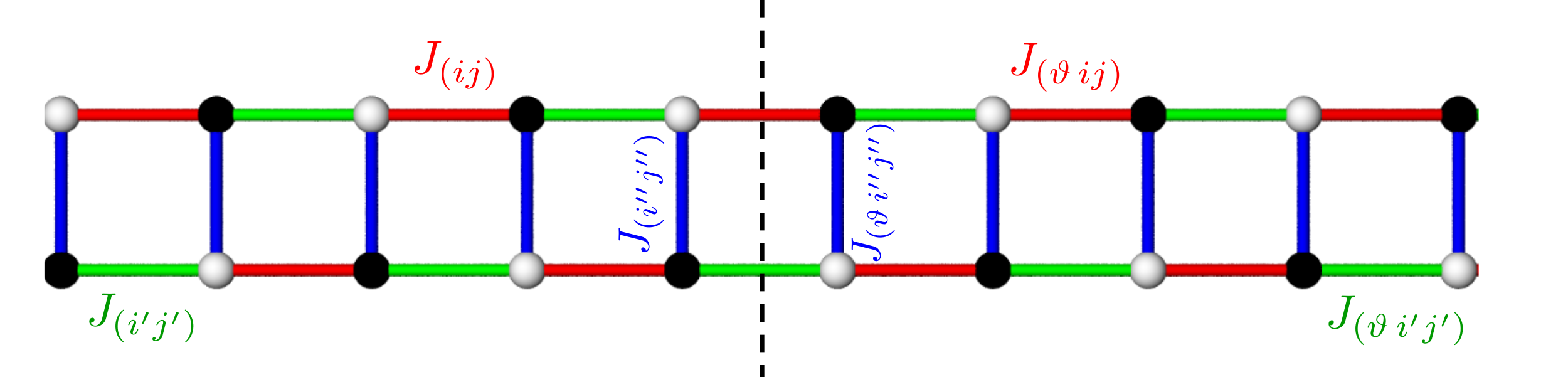}
\caption{Case II: Vertical reflection plane. Recall that $J_{(ij)}=J_{(ji)}$.}
	\label{fig:LadderII}
\end{figure}

\begin{figure}[h!]
	\centering
		\includegraphics[width=0.9\textwidth]{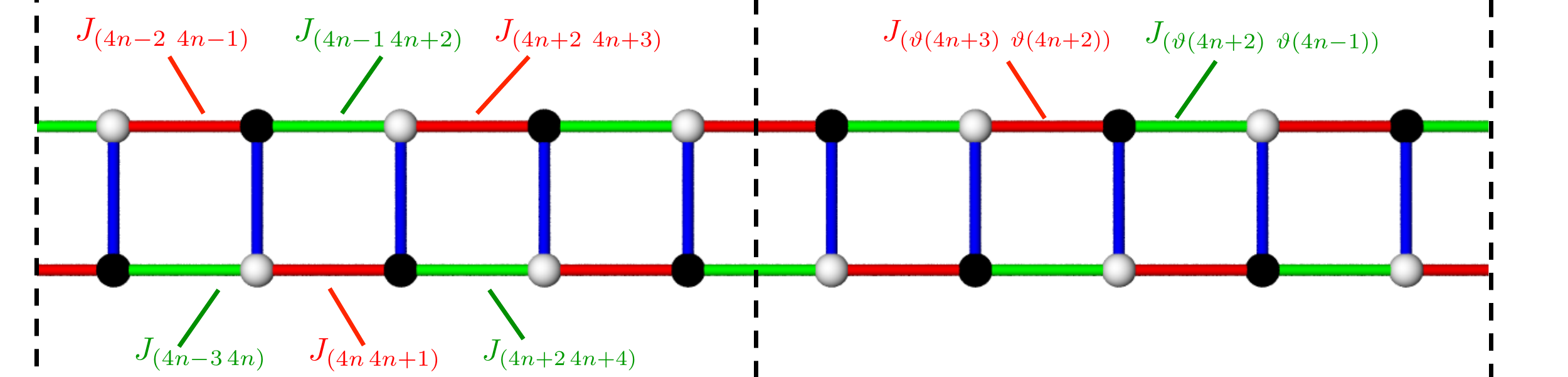}
	\caption{Case III: Vertical reflection of a closed ladder.}
	\label{fig:LadderIII}
\end{figure}

\subsection{Vortex Loops and Reflection-Symmetric Ladders}\label{sec:Ref_symm_case}
In this section we apply Theorem \ref{thm:main_theorem} to characterize the vortex configurations of the ground-state of $\H$, for ladders satisfying conditions \eqref{eq:ref_sym_couplings}.
\begin{thm}\label{prop:Vortex_Closed_Spin_Ladder}
Let $\H$ denote a fermionic ladder Hamiltonian of the form \eqref{eq:HamiltonianExtended} satisfying condition \eqref{eq:ref_sym_couplings} with respect to a reflection $\vartheta$ and a reflection plane $\Pi$. Let $\mathfrak{C}$ denote a non-degenerate, reflection-symmetric loop. Then the vortex configuration of $\mathfrak{C}$ in the ground state of $\H$ is
	\be\label{eq:Theorem_Vortex_Contour}
	\widetilde{{\mathfrak{B}}}(\mathfrak{C})
	=\text{\rm sgn}\lrp{\prod_{\lrp{ij}\in \mathfrak{C}}J_{(ij)}}\,.
	\ee
In case the couplings have all the same sign, then $\widetilde{{\mathfrak{B}}}(\mathfrak{C})=+1$ and the loop $\mathfrak{C}$ is vortex-free.
\end{thm}
\color{black}
\begin{cor}[Vortex-free ladders]\label{rem:vortex-free}
For closed ladders with homogeneous couplings which all have the same sign, every loop $\mathfrak{C}$  is vortex-free in the ground state of $\H$.
\end{cor}

\begin{proof}
Each Hamiltonian $\H$ acting on $\HH$ corresponds to $2^{6N}$ Hamiltonians $\H_{u}$  acting on $\HH_{c}$, some of which could be the same. Each $\H_{u}$ arises from a particular choice of $u_{ij}=\pm1$. The eigenvalues of $\H$ are the union of the eigenvalues of these $2^{6N}$ Hamiltonians $\H_{u}$. A gauge transformation of the variables $b_{j}^{x,y,z}$ transforms one $\H_{u}$ into another $\H_{u'}$. This justifies our present study of the individual Hamiltonians $\H_{u}$.  

Each Hamiltonian $\H_{u}$ is of the form \eqref{eq:u_Hamiltonian}, although the couplings $J_{(ij)}$ may not be positive. In case all the $J_{(ij)}>0$, we infer from Theorem \ref{thm:main_theorem} that the minimum energy of $\H$ is achieved for a $\H_{u}$ with a configuration of the $u_{ij}$'s such that
	\be\label{eq:vortex_lopp_in_proof}
 		\widetilde{\mathfrak{B}}(\mathfrak{C})
		=-\prod_{(ij)\in\mathfrak{C}}\,u_{ij}
		=1\,,
	\ee
for any loop $\mathfrak{C}$ that is reflection-symmetric.
Changing the sign of $J_{(ij)}$ with $(ij)\in\mathfrak{C}$ is equivalent to changing the sign of the corresponding $u_{ij}$, so one infers from  \eqref{eq:vortex_lopp_in_proof} that
	\be
		\widetilde{\mathfrak{B}}(\mathfrak{C})
		=-\prod_{(ij)\in\mathfrak{C}}\,u_{ij}
		=\text{sgn}\lrp{\prod_{(ij)\in\mathfrak{C}}\,J_{(ij)}}\,.
	\ee
This completes the proof of the proposition.  The corollary follows as every  plaquette in the ladder is reflection-symmetric and hence vortex free, and the same then follows for the loop $\mathfrak{C}$. 
$\hfill\qed$
\end{proof}

\subsection{Implications for Reflection-Symmetric Spin Ladders}
For open ladders, we know that the ground-state energies of $\H$ and $H$ agree, as shown in Theorem~\ref{prop:multiplicity}. We also know that the projection $P$ commutes with all the vortex operators, see Proposition \ref{prop:productu}. On the other hand, in the case of a closed ladder we do not know whether the spectra coincide, and in particular whether the ground-state energies are the same. We have shown the following:
\begin{thm}
The ground-state of the Hamiltonian $H$ for an open spin ladder satisfying condition \eqref{eq:ref_sym_couplings} with respect to a reflection plane $\Pi$ has the vortex configuration 
	\be\label{eq:Vortex_Loop_H}
	{\mathfrak{B}}(\mathfrak{C})
	=\text{\rm sgn}\lrp{\prod_{\lrp{ij}\in\mathfrak{C}}J_{(ij)}}\,,
	\ee
in each non-degenerate, reflection-symmetric loop $\mathfrak{C}$ that crosses $\Pi$. In case the couplings have all the same sign, the ground-state is vortex-free in those loops.
\end{thm}

\setcounter{equation}{0}
\section{Numerical Evidence}\label{sec:Numerical_Evidence}
In this section we give some numerical evidence for the spectral properties of $H$ and $\H$, both in the case of open and of closed ladders.  We have  shown in Theorem \ref{prop:multiplicity} that 
 the spectra of $H$ and $\H$ are identical for an open ladder.  However this is not true for a closed ladder. Even a simple closed ladder with $N=2$ (four plaquettes)  shows by explicit numerical diagonalization that $\H$ has eigenvalues not present in the spectrum of $H$, see \S\ref{sec:Numerical Closed Ladders}.  For this Hamiltonian, we plot the energies and show the vortex configurations for a number of eigenvalues. 
 
 We inspect the low-lying spectrum of the Hamiltonians $H$ and $\H$ for a number of ladders of length $N$, in case that  $N$ is as large as $100$, so with up to $400$ spins and $1,600$ Majoranas. We use Mathematica 8.0.4.0 and Matlab 7.10.0.499 (R2010a). In order to find which eigenvalues of $\H$ are eigenvalues of $H$, we use the method introduced in \cite{PCL}.

Our numerical analysis suggests that the ground state of $\H$ and also the ground state of $H$  is vortex free, whether or not they have the symmetry \eqref{eq:ref_sym_couplings}, leading to the conjecture at the end of the section.

\subsection{Open Ladders}\label{sec:numerical_open}
We first analyze an open ladder with $N=2$ (three plaquettes). 
In Fig.~\ref{fig:Numerical_Simulation_5} we plot the low-lying eigenvalues of both  $H$ and $\widetilde{H}$.  
\begin{figure}[h!]
	\centering
		\includegraphics[width=0.87\textwidth]{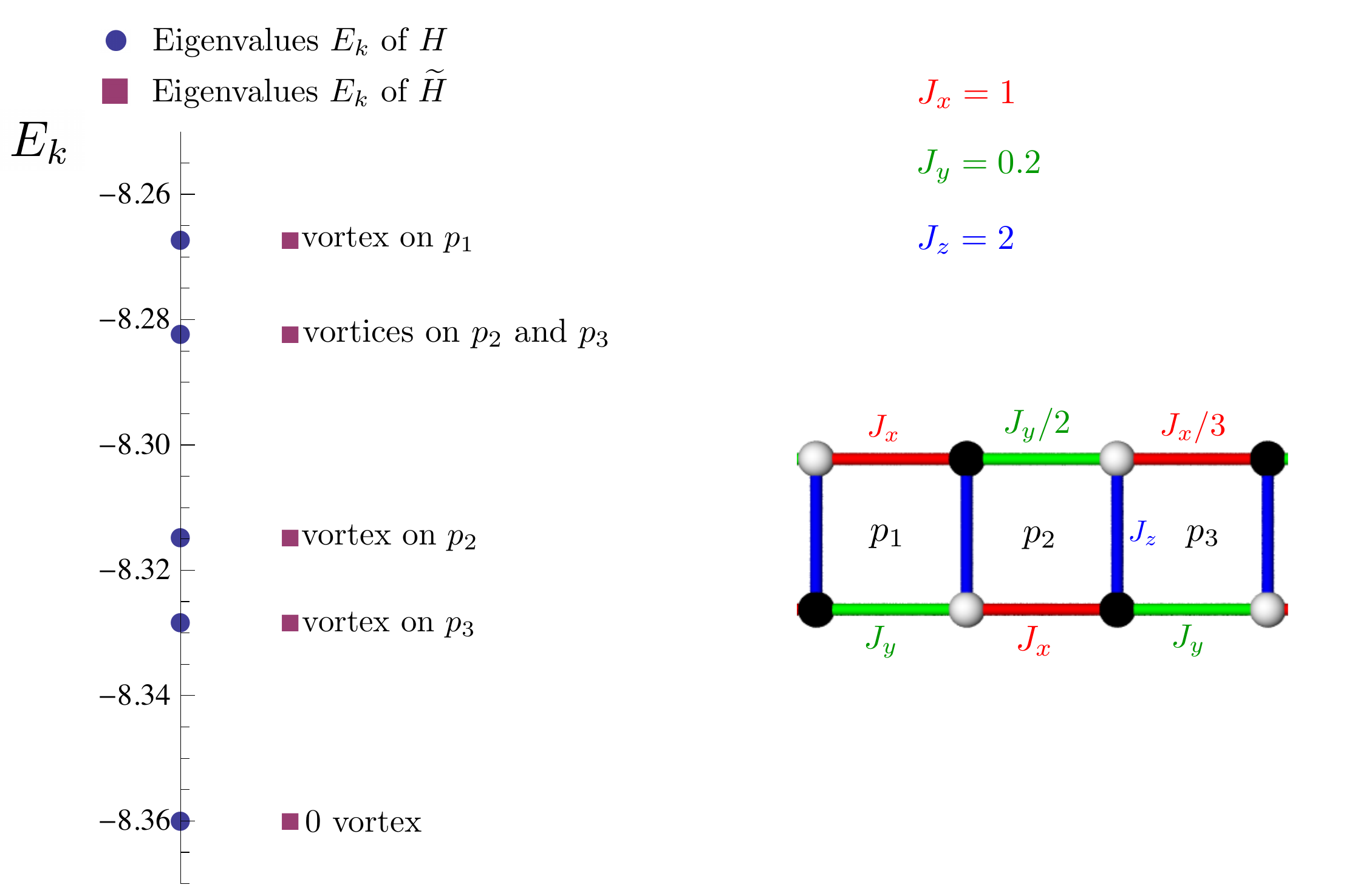}
	\caption{Low-lying eigenvalues of $H$ and $\widetilde H$ for an open ladder with the illustrated couplings. We plot eigenvalues of $H$ with circles and those of $\H$ with squares, and we ignore multiplicities. Other couplings yield qualitatively similar plots.}  
	\label{fig:Numerical_Simulation_5}
\end{figure}
\begin{figure}[h!]
	\centering
		\includegraphics[width=0.9\textwidth]{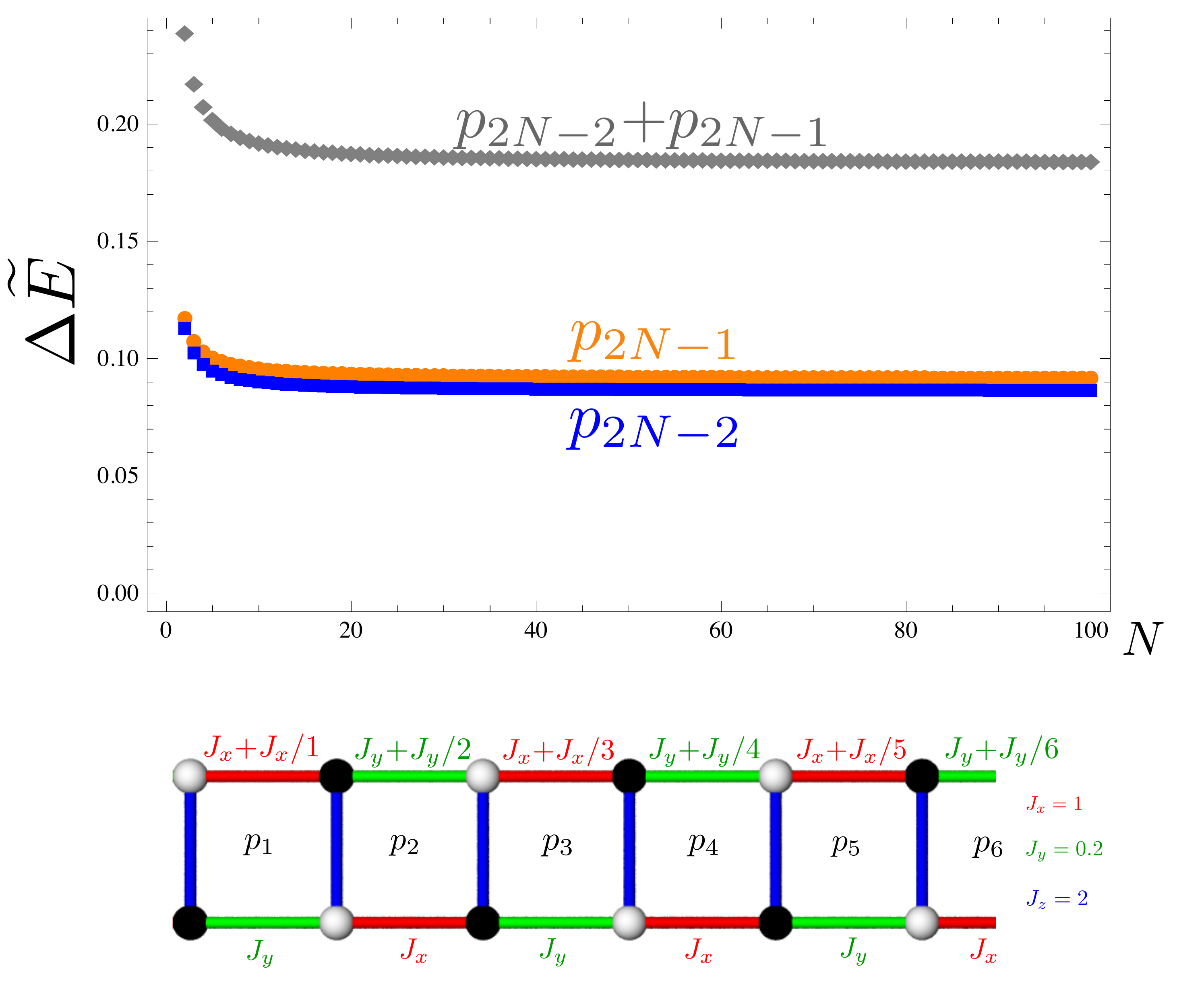}
	\caption{We plot the lowest energies of excitation for an open ladder of length $N$, with a vortex on plaquette $p_{2N-2}$, on plaquette $p_{2N-1}$, or on both. The choice of the couplings is shown in the figure, and the vortex configurations are explained in the text.}
	\label{fig:Vortex_Energy_All_Open}
\end{figure}
We have chosen the couplings $J_{z}=2$, and the $x$ and $y$ couplings to decrease from left to right on the top of the ladder, but not on the bottom. These couplings are neither reflection-symmetric nor homogeneous. The plaquettes are labeled $p_{n}$ with $n=1,2,3$ from left to right as illustrated in Fig.~\ref{fig:Numerical_Simulation_5}.
The numerical eigenvalues of $H$ and $\widetilde{H}$ agree, as we already have shown in Theorem \ref{prop:multiplicity}.  It is interesting that the one-vortex configurations yield the first excited states (aside from multiplicity) and the placement on the ladder  of the vortex   that creates the minimal-energy excitation corresponds to the configuration of coupling constants that one intuitively expects. 

\subsubsection{Hamiltonians for Open Ladders of Length $N$}
Next we consider a sequence of Hamiltonians $\H$ for open ladders with variable length $N$.  We choose non-homogeneous couplings that decay on the upper rungs of the ladder from $2J_{x}$ and $3J_{y}/2$ on the left, to $J_{x}+J_{x}/(2N-1)$ and $J_{y}+J_{y}/(2N-2)$ on the right. On the bottom rungs we take homogeneous couplings. We plot the case $J_{x}=1$, $J_{y}=0.2$, and $J_{z}=2$, as illustrated in Fig.~\ref{fig:Vortex_Energy_All_Open}.
We find that the ground state energy corresponds to a vortex-free configuration.  We then consider the minimal energy excitation above the ground state (neglecting multiplicity).

Among the configuration we have tested, the minimal energy excitation above the vortex-free configuration appears to occur with a single vortex on a plaquette $p_{2N-j}$ for small $j$. The effect of the boundary of the ladder at plaquette $p_{2N-1}$ seems to raise slightly the energy of the single vortex in that plaquette, as illustrated in two curves labeled by $p_{2N-2}$ and $p_{2N-1}$. We have computed other single-vortex excitations that confirm this picture.

We also plot the excitation energy of a configuration with two vortices on plaquettes $p_{2N-2}$ and $p_{2N-1}$. This is approximately twice the energy  of a single vortex.

\subsection{Closed Ladders} \label{sec:Numerical Closed Ladders}
We present numerical evidence for several closed ladders, and contrast the results with the case of the open ladders. In spite of the fact that we observe numerically that $\H$ and $H$ have different spectra, the ground-state energy of $\H$ coincides with the ground-state energy of $H$ and the ground-state vortex-loop configuration is vortex-free.
\begin{figure}[h]
	\centering
		\includegraphics[width=0.9\textwidth]{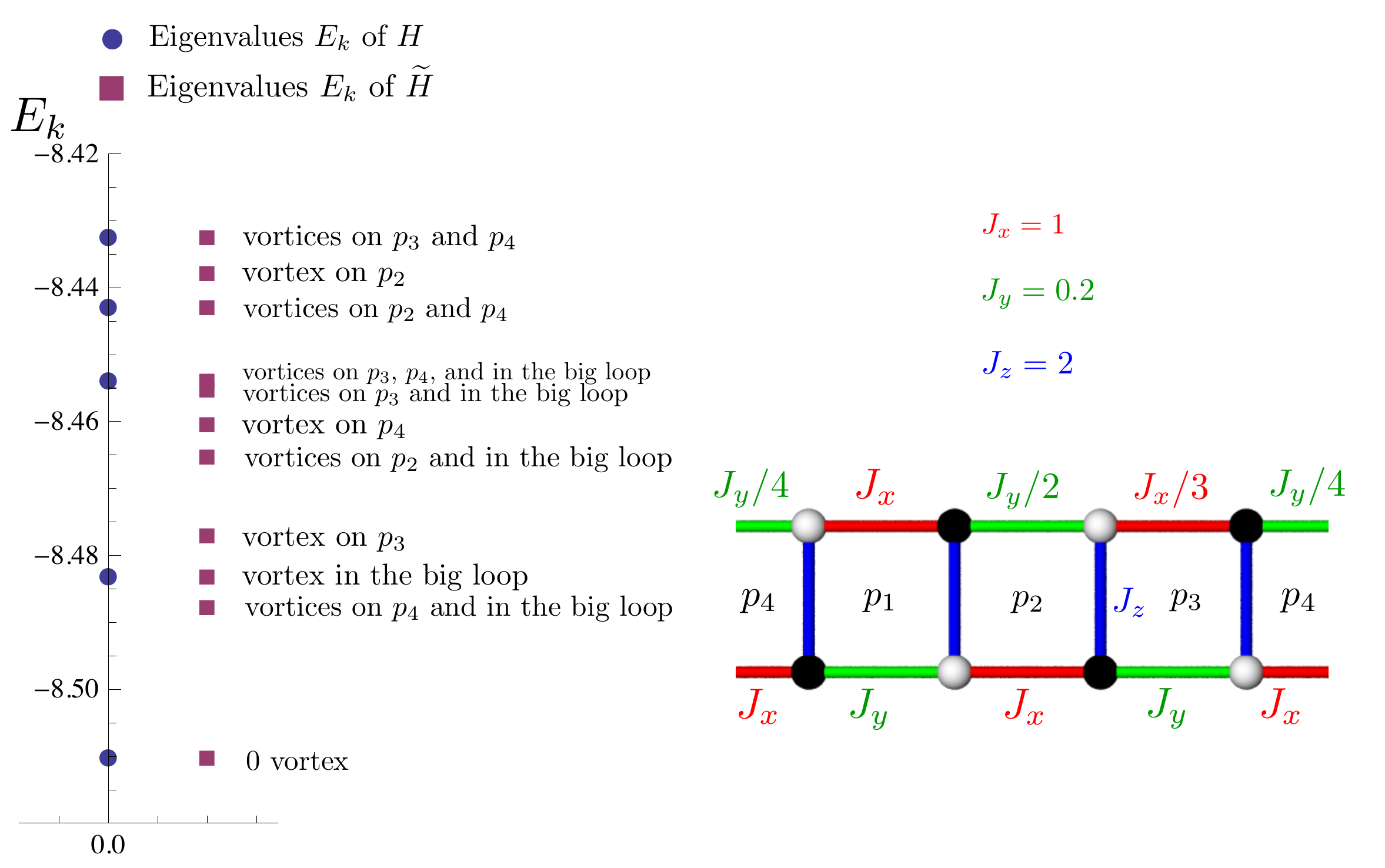}
	\caption{Low-lying eigenvalues of $H$ and $\widetilde H$ for an $N=2$ closed ladder without symmetry. Here $\H$ has eigenvalues that do not occur in $H$. We refer to the ``big loop'' as a loop with four horizontal bonds (around either the top or the bottom of the ladder). As in Fig.~\ref{fig:Numerical_Simulation_5} we ignore multiplicities.}
	\label{fig:Numerical_Simulation_3}
\end{figure}

\subsubsection{Hamiltonians $\H$ and $H$ for Closed Ladders of Length $N=2$}\label{subsec:closed_ladder}
We first analyze the $N=2$ ladder with couplings of the same sort as in Fig. \ref{fig:Numerical_Simulation_5}, but with non-zero couplings on the bonds closing the ladder, as illustrated in Fig.~\ref{fig:Numerical_Simulation_3}. We plot the low-lying eigenvalues of $H$ and $\widetilde{H}$, aside from multiplicity. We label the eigenvalues we plot by their vortex loop configuration. 

\subsubsection{Hamiltonians $\H$ and $H$ for Closed Ladders of Length $N$}\label{sec:closed_ladder_Hamiltonians_Htilde}

Here we consider the two smallest excitations above the ground state of the Hamiltonians $\H$ and $H$ for ladders of variable length $2\leqslant N\leqslant100$. We choose non-homogeneous couplings that decay on the upper rungs of the ladder from $2J_{x}$ and $3J_{y}/2$ on the left, to $J_{x}+J_{x}/(2N-1)$ and $J_{y}+J_{y}/(2N)$ on the right. On the bottom rungs we take homogeneous couplings. We plot the case $J_{x}=1$, $J_{y}=0.2$, and $J_{z}=2$. See Fig.~\ref{fig:Vortex_Energy_All}.

We find that the lowest energy of the configurations we tested is a zero-vortex state. We redefine this energy to be zero. However, we also find that the energy for the state with lowest energy and having a vortex in the big loop, decays rapidly with $N$.  We plot the energy $\Delta \widetilde{E}$ (relative to the vortex-free state) for one vortex in the big loop (BL), two vortices in the big loop and on plaquette $p_{2N}$ (BL+$p_{2N}$), and finally three vortices in the big loop, on $p_{2N-1}$ and $p_{2N}$ (BL+$p_{2N-1}$+$p_{2N}$). The configurations BL and BL+$p_{2N}$ appear to be the lowest-energy excitations of $\H$. By computing the eigenvalues of $H$, we find that the minimal-energy configuration is vortex-free, and the eigenvalue equals the ground state energy of $\H$. However the lowest-energy excitations of $H$ appear to arise from the vortex-loop configurations BL and  BL+$p_{2N-1}$+$p_{2N}$.

\begin{figure}[h!]
	\centering
		\includegraphics[width=0.85\textwidth]{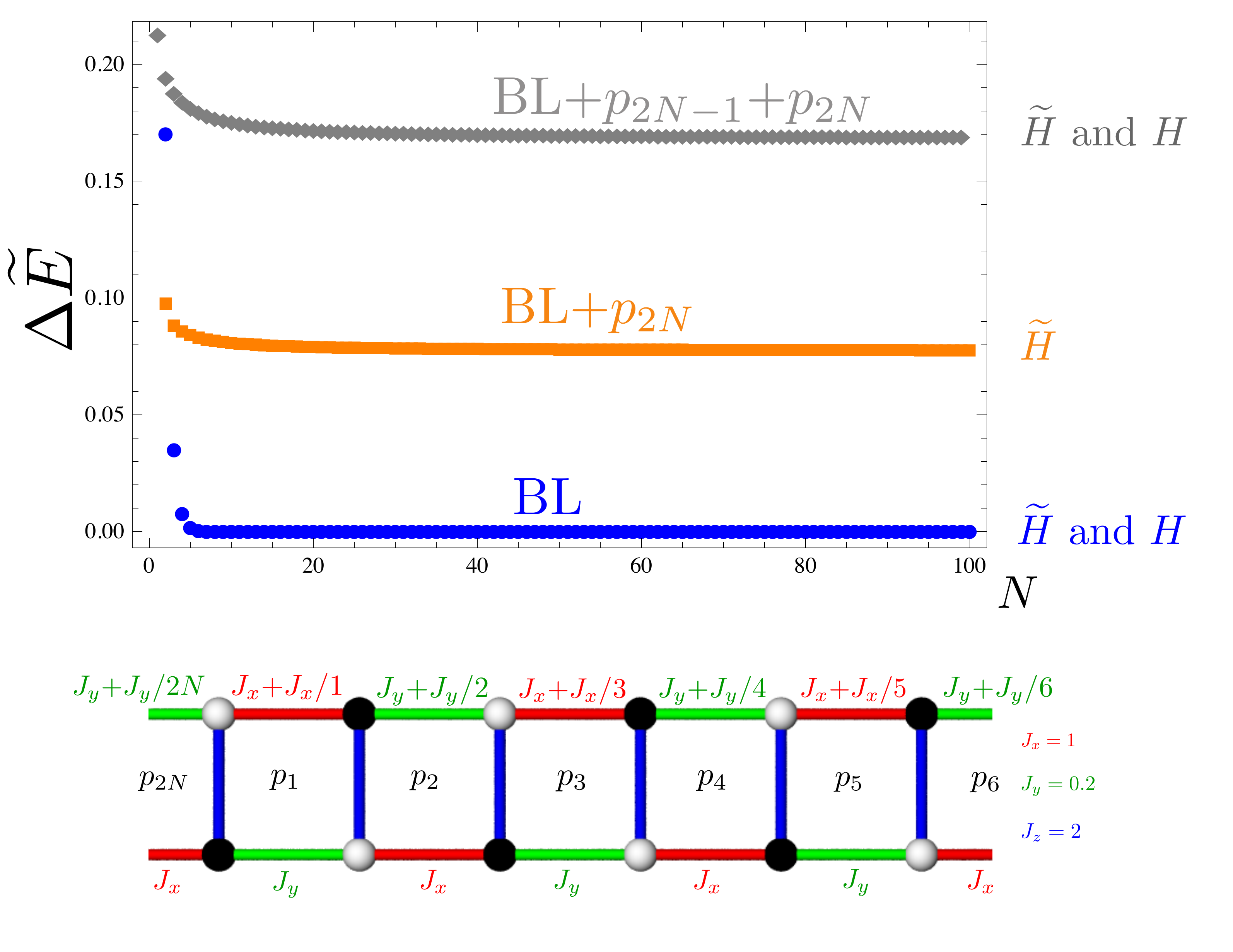}
	\caption{Excitation energies for closed ladders of length $N$ compared with the vortex-free configuration. The choice of the couplings is shown in the figure, and the vortex configurations are explained in the text.}
	\label{fig:Vortex_Energy_All}
\end{figure}

\subsubsection{Remark}
We have performed numerical calculations for different ladder lengths and coupling configurations that we do not show here, but they all result in similar behavior.  

\goodbreak
\setcounter{equation}{0}
\section{Perturbative Results without Reflection Symmetry}\label{sec:non reflection-symm}
\color{red}\color{black}
In the previous sections we found that ladders satisfying \eqref{eq:ref_sym_couplings} and all positive couplings (or all negative couplings) have ground states with no vortex in any reflection-symmetric loop $\mathfrak{C}$. It is of interest to understand whether the vortex-free property extends to open and closed ladders that do not satisfy \eqref{eq:ref_sym_couplings}. Here we investigate this question by perturbation theory, and find evidence that certain ladders have vortex-free ground-state configurations.

We study ladders for which the $x$-couplings are all equal to $J_{x}$, but for which
	\be\label{eq:couplings_perturb}
		J_{x}\gg J_{(i,\,i+3)},\,J_{(i,\,i+1)}\geqslant0\,.
	\ee
For homogeneous couplings with $J_{z}\gg \abs{J_{x}},\abs{J_{y}}$, the ground-state of the open ladder in lowest-order perturbation theory (depending upon $N$) has been shown to be vortex-free when $J_{x}\,J_{y}>0$ or vortex-full when $J_{x}\,J_{y}<0$ \cite{Vishveshwara2011}. Qualitatively this situation is different from the one we study here, as our perturbation satisfying \eqref{eq:couplings_perturb} gives a vortex contribution to the energy only in third order perturbation theory, rather than in second order. On the other hand the perturbation theory evidence in \cite{Vishveshwara2011} that the ground state is vortex-free or vortex-full agrees with Conjecture~\ref{conj:1} in \S\ref{sec:conjecture}.

Write the Hamiltonian as
	\be
		H
		=H_{0}+V\,,
	\ee
where
	\be\label{eq:H0_V}
		H_{0}=-J_{x}\sum_{(ij)_{x}}\sigma_{i}^{x}\sigma_{j}^{x}\,,\quad\text{and}\quad V=-\sum_{(ij)_{y}}J_{(ij)}\,\sigma_{i}^{y}\sigma_{j}^{y}-\sum_{(ij)_{z}}J_{(ij)}\,\sigma_{i}^{z}\sigma_{j}^{z}\,,
	\ee
where $(ij)_{x,y,z}$ denotes type-$x,y,z$ bonds. We consider perturbations of $H_{0}$ by $V$. In the case of the open ladder, $\boldsymbol{\sigma}_{1}$ and $\boldsymbol{\sigma}_{4N}$ do not occur in $H_{0}$.

\begin{prop}[\bf Open Ladder] \label{prop:open ladder perturbation}
Assume that $0<J_{(ij)}$ for all bonds $(ij)$. Also assume that there are constants $0<M_{1}$, $0<M_{2}$ such that $J_{(ij)}<M_{1}$ for $y$ and $z$ bonds $(ij)$ and $M_{2}<J_{x}$. Then for $M_{1}/M_{2}$ sufficiently small, the ground state of Hamiltonian \eqref{eq:Hamiltonian} is vortex-free.
\end{prop}

\smallskip
\noindent {\bf Remark.} We believe that in Proposition \ref{prop:open ladder perturbation} one can choose $M_{1}/M_{2}$ sufficiently small, uniformly in $N$. Establishing such a result about the boundedness of the magnitude of differences of eigenvalues of $H$ requires detailed analysis of the local nature of the perturbation. One needs to estimate non-perturbatively the error in the low-energy perturbation analysis, within a small region of couplings bounded by $M_{1}/M_{2}$, uniformly in $N$. Cluster expansions have been used to do this, both in field theory \cite{GJSAnnals} and in lattice systems. For the latter a framework is given in \cite{Datta1996_1,Datta1996} and several related papers. Working out the details to bound the energy differences for the ladder Hamiltonian $H$ remains an interesting project.   
\begin{proof}

First we establish the notation we use. The ground-state eigenspace $\mathcal{P}_{0}$ of the Hamiltonian $H_{0}$ has $2^{2N+1}$ ground states, which we label by the eigenvalues of $\sigma_{j}^{x}$, for $j=1,\ldots,4N$, with the constraint $\sigma_{i}^{x}\sigma_{j}^{x}=+1$ for all bonds $(ij)_{x}$. We use $m$ to denote the set of eigenvalues of $\sigma_{j}^{x}$ for $j=1,\ldots,4N$ that satisfy the constraint. 
Let $\mathcal{P}_{0}^{\perp}=1-\mathcal{P}_{0}$.
Note that $\mathfrak{B}(\partial p_{n})$ commutes with $H_{0}$ and thus with $\mathcal{P}_{0}$.
Decompose the perturbation $V$ in two parts, $V=V_z+V_{y}$ with 
	\be
		V_{z}
		=\sum_{j=1}^{2N}V_{z}^j=-\sum_{j=1}^{2N}J_{(2j-1,\,2j)}\,\sigma_{2j-1}^{z}\sigma_{2j}^{z}\;,
	\ee 
and 
	\be
		V_{y}=\sum_{j=1}^{2N-1}V_{y}^{j}
		=-\sum_{j=1}^{2N-1}J_{(2j-1,\,2j+2)}\sigma_{2j-1}^{y}\sigma_{2j+2}^{y}\;.
	\ee

\smallskip
\noindent\textbf{The First-Order Effective Hamiltonian.}
The first-order effective Hamiltonian is 
\be\label{eq:SWT_1}
	\mathcal{P}_{0}\,H_{\text{eff}}^{(1)}\,\mathcal{P}_{0}
	= \mathcal{P}_{0}\,V\,\mathcal{P}_{0}
	=0\;.
\ee

\smallskip
\noindent\textbf{The Second-Order Effective Hamiltonian.}
The second-order effective Hamiltonian has matrix elements
\be\label{eq:SWT_2}
(\mathcal{P}_{0}\,H_{\text{eff}}^{(2)}\,\mathcal{P}_{0})_{m,m'}=\frac{1}{2}\sum_{l}V_{m,l}V_{l,m'}\left[\frac{1}{E_m-E_l}+\frac{1}{E_{m'}-E_{l}}\right]\,,\,\,\,\,\,\,\,
\ee
where $l$ labels eigenstates in $\mathcal{P}_{0}^{\perp}$. Here $V_{m,l}$, and $E_{l}$ are the corresponding matrix elements of $V$ and $H_{0}$.
As $\mathcal{P}_{0}\,V_{j}^{z}V_{j'}^{z}\,\mathcal{P}_{0}=0$, $\mathcal{P}_{0}\,V_{j}^{y}V_{j'}^{y}\,\mathcal{P}_{0}=0$ for $j\neq j'$, and $\mathcal{P}_{0}\,V_{j}^{z}V_{j'}^{y}\,\mathcal{P}_{0}=0$ for all $j$ and $j'$, so one obtains
\beq
&&\hskip -0.8cm \mathcal{P}_{0}\,H_{\text{eff}}^{(2)}\,\mathcal{P}_{0}\nn
&=&-\frac{1}{4J_{x}}\left(\sum_{j=1}^{2N-1}J_{(2j-1,\,2j+2)}^{2}+\sum_{j=1}^{2N}J_{(2j-1,\,2j)}^{2}\right.\nn
&&\qquad\left.+J_{(1,\,4)}^{2}+J_{(4N-3,\,4N)}^{2}+J_{(1,\,2)}^{2}+J_{(4N-1,\,4N)}^{2}\right)\mathcal{P}_{0}\,.\nonumber
\eeq
This Hamiltonian does not involve the $\boldsymbol\sigma$'s, so it does not introduce any splitting of the different vortex configurations. 

\smallskip
\noindent\textbf{The Third-Order Effective Hamiltonian.}
The third-order effective Hamiltonian has matrix elements
\begin{eqnarray}\label{eq:SWT_3}
&&(\mathcal{P}_{0}\,H_{\text{eff}}^{(3)}\,\mathcal{P}_{0})_{m,m'}\nn
&&\qquad=-\frac{1}{2}\sum_{l,m''}\left[\frac{V_{m,l}V_{l,m''}V_{m'',m'}}{(E_{m'}-E_l)(E_{m''}-E_l)}+\frac{V_{m,m''}V_{m'',l}V_{l,m'}}{(E_m-E_l)(E_{m''}-E_l)}\right]\\
&&+\frac{1}{2}\sum_{l,l'}V_{m,l}V_{l,l'}V_{l',m'}\left[\frac{1}{(E_m-E_l)(E_m-E_{l'})}+\frac{1}{(E_{m'}-E_l)(E_{m'}-E_{l'})}\right]\,.\nonumber
\end{eqnarray}

We claim this simplifies to 
\beq\label{eq:effective_effective}
	\mathcal{P}_{0}\,H_{\text{eff}}^{(3)}\,\mathcal{P}_{0}
	&=&-\sum_{k=2}^{2N-2}
		\frac{J_{(2k-1,\,2k)}J_{(2k-1,\,2k+2)}J_{(2k+1,\,2k+2)}}
		{8J_{x}^{2}}\,\mathfrak{B}(\partial p_{k})\,\mathcal{P}_{0}\nn
	&&\quad -\frac{J_{(1,\,2)}\,J_{(1,\,4)}\,J_{(3,\,4)}}{2J_{x}^{2}}\mathfrak{B}(\partial p_{1})\,\mathcal{P}_{0}\nn
	&&\quad\quad  -\frac{J_{(4N-3,\,4N-2)}J_{(4N-3,\,4N)}J_{(4N-1,\,4N)}}{2J_{x}^{2}}\mathfrak{B}(\partial p_{2N-1})\,\mathcal{P}_{0}\,.\nn
\eeq

The minimal energy configuration for the Hamiltonian \eqref{eq:effective_effective} therefore occurs in the case that all $\mathfrak{B}(\partial p_{k})=+1$. The single sum over $k$ reflects the extensive nature of the eigenvalues in perturbation theory, see for example \cite{Bravyi_SW_2011}. The splitting of the degenerate ground states occurs in case a single vortex $\mathfrak{B}(\partial p_{k})=-1$. This raises the energy of such a state by the quantity 
	\be
		\delta E
				=\left\{\begin{array}{cc}
				\frac{J_{(2k-1,\,2k)}J_{(2k-1,\,2k+2)}J_{(2k+1,\,2k+2)}}{4J_{x}^{2}}\,,\hfill&\text{for } k=2,\ldots,2N-2\hfill\\
				\frac{J_{(1,\,2)}\,J_{(1,\,4)}\,J_{(3,\,4)}}{J_{x}^{2}}\,,\hfill& \text{for } k=1\hfill\\
				\frac{J_{(4N-3,\,4N-2)}J_{(4N-3,\,4N)}J_{(4N-1,\,4N)}}{J_{x}^{2}}\,, \hfill& \text{for } k=2N-1\hfill
				 \end{array} \right.\,.
	\ee
Which plaquette $p_{k}$ gives rise to the minimal energy shift depends upon the choice of the coupling constants $J_{(ij)}$. In every case, the energy shift is positive as long as $J_{(ij)}>0$. For given $M_{1}$ and $M_{2}$, the energy shifts $\delta E$ due to a single vortex on one plaquette---as given by third-order perturbation theory---are bounded away from zero, and also from above, uniformly in $N$.

We justify the expression \eqref{eq:effective_effective} as follows. The first sum in \eqref{eq:SWT_3} vanishes because $\mathcal{P}_{0}V\mathcal{P}_{0}=0$.
The perturbation $V_{z}^{j}$ contains the product $\sigma_{2j-1}^{z}\sigma_{2j}^{z}$ and $V_{y}^j$ contains the product $\sigma_{2j-1}^{y}\sigma_{2j+2}^{y}$, so the only possible  third-order terms have the form $V_{z} V_{y} V_{z}$, $V_{z} V_{z} V_{y}$, or $V_{y} V_{z} V_{z}$, where
\beq
V_{z} V_{y} V_{z}&=&- \sum_{j,k,l}J_{(2j-1,\,2j)}J_{(2k-1,\,2k+2)}J_{(2l-1,\,2l)}\,\sigma_{2j-1}^{z}\sigma_{2j}^z\,\sigma_{2k-1}^{y}\sigma_{2k+2}^{y}\,\sigma_{2l-1}^{z}\sigma_{2l}^{z}\,,\nonumber
\eeq
etc.
There are only two possible choices of indices such that $\mathcal{P}_{0}\,V_{z} V_{y} V_{z}\,\mathcal{P}_{0}$ does not vanish, namely $j=k$, $l=k+1$, and $l=k$ and $j=k+1$. One thus obtains
\begin{eqnarray}\label{eq:zyz}
&&\hskip -1cm \mathcal{P}_{0}\,V_{z} V_{y} V_{z}\,\mathcal{P}_{0}\nn
&=&-2\sum_{k}J_{(2k-1,\,2k)}J_{(2k-1,\,2k+2)}J_{(2k+1,\,2k+2)}\mathcal{P}_{0}\,\sigma_{2k-1}^{x}\,\sigma_{2k}^{z}\sigma_{2k+1}^{z}\sigma_{2k+2}^{x}\,\mathcal{P}_{0}\nonumber\\
&=&2\sum_{k}J_{(2k-1,\,2k)}J_{(2k-1,\,2k+2)}J_{(2k+1,\,2k+2)}\,\mathfrak{B}(\partial p_{k})\,\mathcal{P}_{0}\,.
\end{eqnarray}
Here we use $\mathcal{P}_{0}\,\sigma_{2k}^{z}\sigma_{2k+1}^{z}\,\mathcal{P}_{0}=-\mathcal{P}_{0}\,\sigma_{2k}^{y}\sigma_{2k+1}^{y}\,\mathcal{P}_{0}$ and the definition of $\mathfrak{B}(\partial p_{k})$ in \eqref{eq:Vortex Loop Operator}.
Similarly
\begin{eqnarray}\label{eq:zzy}
&&\hskip -1cm \mathcal{P}_{0}\,V_{z} V_{z} V_{y}\,\mathcal{P}_{0}\nn
&=&2\sum_{k}J_{(2k-1,\,2k)}J_{(2k-1,\,2k+2)}J_{(2k+1,\,2k+2)}\,\mathcal{P}_{0}\sigma_{2k-1}^{x}\sigma_{2k}^{z}\sigma_{2k+1}^{z}\sigma_{2k+2}^{x}\mathcal{P}_{0}\nonumber\\
&=&-2\sum_{k}J_{(2k-1,\,2k)}J_{(2k-1,\,2k+2)}J_{(2k+1,\,2k+2)}\,\mathfrak{B}(\partial p_{k})\,\mathcal{P}_{0}\,.
\end{eqnarray}

The terms in \eqref{eq:zyz}--\eqref{eq:zzy} that do not contain the boundary plaquettes $\mathfrak{B}(\partial p_{1})$ and $\mathfrak{B}(\partial p_{2N-1})$ cancel identically; they have the same energy denominators  and opposite signs. Finally
\begin{eqnarray}
&&\hskip -1cm \mathcal{P}_{0}V_{y} V_{z} V_{z}\,\mathcal{P}_{0}\nn
&=&2\sum_{k}J_{(2k-1,\,2k)}J_{(2k-1,\,2k+2)}J_{(2k+1,\,2k+2)}\,\mathcal{P}_{0}\sigma_{2k-1}^{x}\sigma_{2k}^{z}\sigma_{2k+1}^{z}\sigma_{2k+2}^{x}\,\mathcal{P}_{0}\nonumber\\
&=&-2\sum_{k}J_{(2k-1,\,2k)}J_{(2k-1,\,2k+2)}J_{(2k+1,\,2k+2)}\,\mathfrak{B}(\partial p_{k})\,\mathcal{P}_{0}\,.
\end{eqnarray}
Therefore the contribution to $\mathcal{P}_{0}H_{\text{eff}}^{(3)}\mathcal{P}_{0}$ that does not involve the boundary plaquettes $p_{1}$ and $p_{2N-1}$ is
$$-\sum_{k=2}^{2N-2}
		\frac{J_{(2k-1,\,2k)}J_{(2k-1,\,2k+2)}J_{(2k+1,\,2k+2)}}
		{8J_{x}^{2}}\,\mathfrak{B}(\partial p_{k})\,\mathcal{P}_{0}\,.$$

The situation is different for terms entering in the perturbations $V_{z}V_{y}V_{z}$, $V_{z}V_{z}V_{y}$, and $V_{y}V_{z}V_{z}$ and involving plaquettes $p_{1}$ and $p_{2N-1}$. The reason is that $\boldsymbol\sigma_{1}$ and $\boldsymbol\sigma_{4N}$ do not enter into $H_{0}$. 
Taking this into account, the coefficients of the boundary terms differ. However, they are still negative and the third-order effective Hamiltonian is \eqref{eq:effective_effective}.

The fact that the perturbation theory result applies in a region of couplings for small $M_{1}/M_{2}$ is a consequence of the analyticity of the eigenvalues, see \cite{Kato1949} and \S II.1.3--II.1.4 of \cite{KatoBook} .
$\hfill\qed$
\end{proof}

\begin{prop}[\bf Closed Ladder with $N>2$]
Under the hypothesis of Proposition \ref{prop:open ladder perturbation}, the ground state of the Hamiltonian $H$ in \eqref{eq:Hamiltonian} with closed boundaries is vortex-free on each plaquette $p_{1},\ldots,p_{2N}$. The effective Hamiltonian to third order is
	\beq
		\mathcal{P}_{0}H_{\text{eff}}^{\leqslant 3}\mathcal{P}_{0}
		&=&H_{0}\mathcal{P}_{0}
		-\frac{1}{4J_{x}}\lrp{\sum_{j=1}^{2N-1}J_{(2j-1,\,2j+2)}^{2}+\sum_{j=1}^{2N}J_{(2j-1,\,2j)}^{2}+J_{(2,\,4N-1)}^{2}}\,\mathcal{P}_{0}\nn
		&&-\sum_{k=1}^{2N-1}\frac{J_{(2k-1,\,2k)}J_{(2k-1,\,2k+2)}J_{(2k+1,\,2k+2)}}{8J_{x}^{2}}\,\mathfrak{B}(\partial p_{k})\,\mathcal{P}_{0}\,.
	\eeq
\end{prop}

\noindent \textbf{Remark.} The perturbative expansion up to third order will not give a splitting in energy due to a vortex on the big loop (the shortest loop around the top or bottom of the closed ladder). This will occur only in perturbation theory of order $O(N)$; but a single vortex in this loop gives an energy shift that is exponentially small in $N$.

\begin{proof} 
The Hamiltonian $H$ in \eqref{eq:Hamiltonian} possesses two additional bonds $(2,\,4N-1)$ and $(1,\,4N)$ that do not occur in the open ladder. They yield interaction terms
	\be
	H_{x}^{\text{closed}}=-J_{x}\,\sigma_{1}^{x}\,\sigma_{4N}^{x}\qquad			\text{and}\qquad V_{y}^{\text{closed}}=-J_{(2,\,4N-1)}\,\sigma_{2}^{y}					\sigma_{4N-1}^{y}\,.
	\ee
We incorporate $H_{x}^{\text{closed}}$ in the unperturbed Hamiltonian $H_{0}$ and $V_{y}^{\text{closed}}$ in the perturbation $V_{y}$. Define $V_{y}^{2N}=V_{y}^{\text{closed}}$, we have $V_{y}=\sum_{j=1}^{2N}V_{y}^{j}$. We now derive the first order, second order, and third order effective Hamiltonians.

\smallskip
\noindent \textbf{The First-Order Effective Hamiltonian.}  As in the  proof of Proposition \ref{prop:open ladder perturbation}, the first-order effective Hamiltonian vanishes.

\smallskip
\noindent \textbf{The Second-Order Effective Hamiltonian.} Also as in the  proof of Proposition \ref{prop:open ladder perturbation}, the second-order is given in \eqref{eq:SWT_2}. For ladders with $N>2$ the only second order terms that do not vanish are $\mathcal{P}_{0}\lrp{V_{z}^{j}}^{2}\mathcal{P}_{0}$ and $\mathcal{P}_{0}\lrp{V_{y}^{j}}^{2}\mathcal{P}_{0}$. The ladder being closed, all the energy denominators in \eqref{eq:SWT_2} are the same. One thus obtains the second-order term.

\smallskip
\noindent \textbf{The Third-Order Effective Hamiltonian.} As for the open ladder, the third-order effective Hamiltonian is given in \eqref{eq:SWT_3}. For the same reason as in the case of the open ladder, the first sum in \eqref{eq:SWT_3} vanishes. Again the relevant perturbations are $\mathcal{P}_{0}V_{z} V_{y} V_{z}\mathcal{P}_{0}$, $\mathcal{P}_{0}V_{z} V_{z} V_{y}\mathcal{P}_{0}$, $\mathcal{P}_{0}V_{y} V_{z} V_{z}\mathcal{P}_{0}$. The cancelation of the terms \eqref{eq:zyz} and \eqref{eq:zzy} of the first two perturbations for the open ladder also takes place for the closed ladder. Furthermore, since the ladder is closed, the energy denominators appearing in $\mathcal{P}_{0}V_{z} V_{y} V_{z}\mathcal{P}_{0}$, $\mathcal{P}_{0}V_{z} V_{z} V_{y}\mathcal{P}_{0}$, $\mathcal{P}_{0}V_{y} V_{z} V_{z}\mathcal{P}_{0}$ are all the same and no ``boundary'' terms appear in the third-order effective Hamiltonian.
$\hfill\qed$
\end{proof}

\section{Conjecture}\label{sec:conjecture}
Based on the numerical calculations that we performed in \S\ref{sec:Numerical_Evidence} and the perturbation calculations we performed in \S\ref{sec:non reflection-symm}, we formulate the following conjecture for ladder Hamiltonians:

\begin{conj}\label{conj:1}
For a closed ladder, the ground state energies  of $H$ in \eqref{eq:Hamiltonian} and $\widetilde{H}$ in \eqref{eq:HamiltonianExtended} coincide.   For a closed or open ladder with the coupling constants $J_{(ij)}$  all positive or all negative, the ground states of $H$ and $\H$ are vortex-free.
\end{conj}

It is known that one has qualitatively different behavior in two-dimensional spin systems with trivalent interactions \cite{PCL}.

\section{Acknowledgements}
Arthur Jaffe wishes to thank Daniel Loss for hospitality at the University of Basel where most of this work has been done, to Gian Michele Graf for hospitality at the ETH, and to J\"urg Fr\"ohlich for discussions. This work was supported by the Swiss NSF, NCCR QSIT, NCCR Nanoscience, 
NSERC, CIFAR, FRQNT, INTRIQ, 
and the Pauli Center of the ETH.

%\section{Introduction}
%\label{intro}
%Your text comes here. Separate text sections with
%\section{Section title}
%\label{sec:1}
%and \cite{RefJ}
%\subsection{Subsection title}
%\label{sec:2}
%as required. Don't forget to give each section
%and subsection a unique label (see Sect.~\ref{sec:1}).
%%
%% For one-column wide figures use
%\begin{figure}
%% Use the relevant command for your figure-insertion program
%% to insert the figure file.
%% For example, with the option graphics use
%\centering
%\resizebox{0.75\textwidth}{!}{%
%  \includegraphics{leer.eps}
%}
%% If not, use
%%\vspace{5cm}       % Give the correct figure height in cm
%\end{figure}
%\begin{table}
%\centering
%\begin{tabular}{lll}
%\hline\noalign{\smallskip}
%\end{tabular}
%% Or use
%%\vspace*{5cm}  % with the correct table height
%\end{table}
%%
% BibTeX users please use
% \bibliographystyle{}
% \bibliography{}

\end{document}